\documentclass[10pt,letterpaper]{article}
\usepackage{arxiv}
\usepackage[utf8]{inputenc}
\usepackage[T1]{fontenc}
\usepackage{graphicx,booktabs,amsmath,amssymb,amsthm,nicefrac,microtype}
\usepackage[numbers,sort&compress]{natbib}
\usepackage[hyphens]{url}
\usepackage{algorithm,algpseudocode,subcaption,longtable}
\usepackage{hyperref}
\usepackage{doi}
\usepackage{array,tabularx,float,xcolor,placeins}
\DeclareUnicodeCharacter{2212}{\ensuremath{-}}
\DeclareUnicodeCharacter{2011}{-}
\newtheorem{lemma}{Lemma}[section]
\newtheorem{prop}{Proposition}

\providecommand{\Description}[1]{}
\title{Incentivizing User Data Contributions for LLM Improvement under Withdrawal Rights}
\author{%
\normalfont\small
\begin{tabular}{@{}ccc@{}}
  \begin{minipage}[t]{0.30\textwidth}
    \centering
    \textbf{Di Feng}\\
    School of Finance\\
    Dongbei University of\\
    Finance and Economics\\
    \texttt{fengdi@dufe.edu.cn}
  \end{minipage}
  &
  \begin{minipage}[t]{0.30\textwidth}
    \centering
    \textbf{Runyuan Zhao}\\
    School of Artificial Intelligence\\
    The Chinese University of\\
    Hong Kong, Shenzhen\\
    \texttt{theazhao421@gmail.com}
  \end{minipage}
  &
  \begin{minipage}[t]{0.30\textwidth}
    \centering
    \textbf{Chenhao Zhang}\\
    School of Artificial Intelligence\\
    The Chinese University of\\
    Hong Kong, Shenzhen\\
    \texttt{20003112@mail.ecust.edu.cn}
  \end{minipage}
  \\ \noalign{\vskip 0.65em}
  \multicolumn{3}{c}{%
    \begin{minipage}[t]{0.65\textwidth}
      \centering
      \textbf{Zhanzhan Zhao}\thanks{Corresponding author.}\\
      School of Humanities and Social Science\\
      School of Artificial Intelligence\\
      The Chinese University of Hong Kong, Shenzhen\\
      \texttt{zhanzhanzhao@cuhk.edu.cn}
    \end{minipage}%
  }
\end{tabular}%
}
\hypersetup{pdfauthor={Di Feng, Runyuan Zhao, Chenhao Zhang, Zhanzhan Zhao}}

\date{}
\renewcommand{\shorttitle}{Incentivizing User Data Contributions for LLM Improvement}
\hypersetup{hidelinks,pdftitle={Incentivizing User Data Contributions for LLM Improvement under Withdrawal Rights},pdfkeywords={Mechanism Design; Data Contribution Incentives; Withdrawal Rights; Threshold Public Goods; Large Language Models; Agent-Based Simulation}}
\begin{document}
\maketitle
\begin{abstract}
As large language models (LLMs) become more prevalent, they raise concerns about the erosion of the public knowledge on which their development depends. Users increasingly turn to LLMs for information, potentially reducing their engagement with public knowledge-sharing communities and their willingness to contribute knowledge themselves. This creates a dilemma for platforms seeking to improve LLMs: continued improvement depends on user contributions, yet contributing entails privacy and effort costs, while relying on existing models or others’ contributions offers a less costly alternative.  We study this challenge through a threshold public-goods model with heterogeneous user costs. Under uniform per-unit subsidies, equilibrium behavior exhibits subsidy leakage: users receive compensation even when aggregate contributions fall short of the improvement threshold. To address this inefficiency, we develop a mechanism combining verifiable cost disclosure, personalized contribution assignments, and voluntary withdrawal. Equilibrium analysis establishes that rational assignment yields either sustainable, profitable provision or no collection. We further identify a trade-off in withdrawal timing: simultaneous decisions minimize contribution costs whenever both arrangements succeed, whereas sequential decisions ordered by increasing assignments allow later users to cover shortfalls and sustain provision under broader cost conditions. Numerical experiments quantify the effects of these institutional choices on provision probability and social welfare. We also design paired simulations with agents grounded in anonymized ACS PUMS profiles, replicating each agent across institutional settings to compare contribution and withdrawal behavior during interactions with a target LLM under matched conditions. Together, our analyses clarify how contribution incentives, personalized assignments, and withdrawal rights shape the collective provision of data for LLM improvement, offering a framework for platform mechanisms that support data collection while preserving user control.
\end{abstract}

\section{Introduction}
Large language models (LLMs) are increasingly embedded in personal and professional workflows, supporting tasks such as drafting emails, revising documents, seeking information, learning new skills, and developing software~\citep{chatterji2025how}. 
In these interactions, users can supply task-specific context, contribute original material, correct inaccurate responses, and express preferences among alternative outputs. 
Users therefore act not only as consumers of AI services but also as potential contributors of training examples and feedback, subject to consent and data quality requirements~\citep{christiano2017deep,stiennon2020learning,ouyang2022training}.
This dual role creates the following training-data dilemma: continued LLM improvement requires high-quality and novel human-generated data beyond existing public corpora \citep{hoffmann2022training,grattafiori2024llama}, yet increasing reliance on LLMs may reduce users' own production of original content \citep{villalobos2024position,luo2025survey}.
Moreover, valuable data contributions require effort, create privacy and proprietary risks, and often yield no direct compensation~\citep{carlini2021extracting,acquisti2016economics,acemoglu2022too}. 
An LLM provider thus needs sufficient user-supplied data to improve a
shared model, while individual users have incentives to withhold contributions whose costs are private and whose benefits are shared broadly.

To translate this dilemma into a tractable strategic setting, we focus on two features of LLM data provision. First, it has a threshold public-good structure: each user privately pay their own contribution cost, while model improvement benefits the broader population only when aggregate data reach a collective target \citep{bagnoli1991voluntary,marx2000dynamic,sandler2015collective}.
Such a jumping target may arise from a discontinuous technical improvement, or may emerge because continued training creates substantial practical value \citep{ouyang2022training,grattafiori2024llama}, emergent capabilities is a leading example for this \citep{wei2022emergent}.

Second, data contributions are potentially reversible. Users can delete interaction histories, revoke access, or restrict data usage through technical mechanisms such as machine unlearning, as well as through regulatory frameworks like data protection laws \citep{ginart2019making,guo2019certified}. In particular, the GDPR (General Data Protection Regulation) allows users to withdraw consent at any time and provides a qualified right to erasure \citep{gdpr2016}. In the LLM data provision setting, users may withdraw their assigned contributions. This introduces dynamic strategic behavior not present in standard one-shot public-goods formulations.

These observations raise two central questions. Can an LLM provider coordinate
heterogeneous users toward a collective data target while preserving withdrawal rights?
How does withdrawal timing affect provision robustness and social welfare? We address these questions by comparing alternative institutional rules for contribution assignment and withdrawal.

Section~\ref{sec:model} develops a stylized model in which users incur heterogeneous private costs and receive a common benefit when retained contributions reach the target. The provider offers a uniform subsidy, elicits verifiable costs, assigns personalized contribution targets, and permits withdrawal before training.
We use $p=0$  to represent the baseline without explicit per-unit compensation and
and $p>0$ to represent cash or other rewards. Abstracting from specific technical processes allows us to isolate the role of institutional rules. Formal analysis characterizes the equilibria and their conditions, while simulations illustrate the resulting provision and welfare outcomes.

To complement the theoretical analysis and equilibrium simulations, in Section~\ref{sec:experiments} we conduct an
exploratory dialogue experiment using LLM-based agents grounded in $50$ anonymized
profiles from the 2018 American Community Survey (ACS)~\citep{uscensus2018pums}. Each profile is evaluated under all three mechanisms, yielding $150$ paired trajectories. This within-profile design compares contribution costs, provision success, and modeled social welfare under matched conditions and illustrates how contribution and withdrawal decisions unfold during interactions with a target LLM. 

Our approach differs from existing work in both its object of study and its modeling purpose. Formal economic and mechanism-design research on LLMs has examined auctions for AI-generated content and the pricing of model access, token allocation, and fine-tuning \citep{duetting2024mechanism,bergemann2025llm}. Research on collaborative learning, federated-learning incentives, and data markets studies compensation, data valuation, privacy costs, procurement, and optimal contracts \citep{sim2020collaborative,sim2023incentives,zhan2021survey,chen2022selling,zhang2024datamarkets,zhang2025fairshare,fan2025data,wang2025paid}. 
Whereas this work primarily focuses on pricing or valuing individual contributions, We instead take a uniform subsidy mechanism as given and study how a provider can coordinate multiple users around a collective threshold when assigned contributions remain voluntarily withdrawable.
Our analysis also connects threshold public-good models with research on data control. The former explains how beliefs, strategic uncertainty, heterogeneity, and the timing of initial contributions affect collective provision \citep{palfrey1984participation,bagnoli1991voluntary,croson2000step,vanhuyck1990tacit,compte2003voluntary,marx2000dynamic,varian1994sequential,bag2011sequential}; the latter studies privacy costs, data deletion, and the technical removal of data from trained models \citep{acquisti2016economics,carlini2021extracting,ginart2019making,guo2019certified}. 
Our distinction concerns the timing of the strategic decision: rather than
studying the timing of initial contributions or post-training data removal, we examine
withdrawal after a provider has observed disclosed costs and assigned contribution
targets.

Our contributions are threefold. First, Section~\ref{subs} identifies \emph{subsidy leakage}: pessimistic expectations can lead users to contribute only the amounts justified by the subsidy, leaving aggregate data below the target even though the provider pays for their contributions. Because the same subsidy can also sustain successful provision under more favorable expectations, the failure reflects a coordination problem rather than simply an insufficient payment. 

Second, Section~\ref{sec:withdrawal} introduces a mechanism combining verifiable cost disclosure, personalized assignment, and voluntary withdrawal. Under rational assignment, 
the provider makes a positive assignment only when users have incentives to retain enough data and successful provision is profitable; otherwise, it requests no data. Assignment therefore coordinates users around a feasible target and replaces unsuccessful paid collection with either successful model improvement or no collection. 

Third, Section~\ref{sec:comparison} 
establishes a trade-off between cost efficiency and provision robustness.
Whenever both protocols succeed, simultaneous withdrawal minimizes total privacy and effort cost among the feasible assignments considered. Small-first sequential withdrawal instead allows later users to observe earlier users' retained contributions before making their own decisions. These observed commitments reduce uncertainty about whether the collective target will be reached, allowing sequential withdrawal to sustain provision in some cases where simultaneous decisions fail.
Section~\ref{sec:experiments} further provides numerical evidence for
this trade-off, and the exploratory dialogue experiment offers a complementary
illustration: simultaneous withdrawal has lower contribution costs when both protocols succeed, whereas small-first withdrawal achieves higher provision success and mean social welfare in the evaluated settings.

These results show that data incentives should be evaluated by whether they produce usable model improvement, not merely participation or data volume. They also show that withdrawal rights need not obstruct collective AI development: when incorporated into the mechanism, withdrawal can preserve user control while preventing users and providers from bearing costs for a training campaign that cannot succeed.

\section{Model}
\label{sec:model}

Consider an LLM platform with $n$ users and one provider seeking to continue training a shared model. Each user can contribute data to the adaptation campaign. We represent the campaign as a threshold technology: although model performance may vary continuously with the amount of training data, the continued training generates a meaningful collective benefit only when the total amount of usable data reaches a predefined target. The provider determines the contribution threshold \(X\) before the campaign and publicly announces it to all users. We therefore treat \(X\) as fixed and commonly known throughout the analysis.

The provider offers a uniform per-unit subsidy $p \geq 0$ for effective data contributions. Let $e_i \in [0,1]$ denote user $i$'s final effective contribution, measured as a normalized amount of usable data, where $e_i=1$ is the maximum feasible contribution. How the final contribution $e_i$ is determined depends on the mechanism considered below.

Each user $i$ has a private cost parameter $c_i$, independently drawn from a common distribution $F$ on $[\underline c,\bar c]$, where $0<\underline c<\bar c<\infty$ and $F$ has a continuous positive density $f$. The distribution $F$ is common knowledge, while each user observes only their own realized cost. We interpret $c_i$ as the user's private cost of making data available, including privacy and effort costs associated with data sharing~\citep{fan2025data}. For analytical tractability, we assume a quadratic contribution cost $c_i e_i^2/2$, which captures increasing marginal costs. 

Let $G=\sum_{i=1}^n e_i$ denote the aggregate final effective contribution. We represent the common benefit generated by the adaptation campaign as

$$
H(G)=
\begin{cases}
V, & \text{if } G\geq X,\\
0, & \text{if } G<X,
\end{cases}
$$

where $V>0$ is the benefit that each user receives when the predefined quality target is met. We refer to the event $G\geq X$ as successful provision. This benefit is available to all users, including those who do not contribute data. We assume that $X\in(m,m+1)$ for some integer $m\geq1$ and that $n\geq m+1$. Because each user can contribute at most one unit, no group of at most $m$ users can reach the target. Successful adaptation therefore requires positive contributions from at least $m+1$ users, making provision a collective rather than an individual outcome.

We assume quasi-linear utilities. Given a final contribution profile $\boldsymbol{e}=(e_1,\ldots,e_n)$, user $i$'s utility is

\begin{equation}
\label{eq:ui}
U_i
=
V\mathbf{1}_{\left\{\sum_j e_j\geq X\right\}}
+pe_i
-\frac{c_i e_i^2}{2}.
\end{equation}

The three terms represent, respectively, the common benefit from meeting the model-quality target, the subsidy received for the user's final effective contribution, and the user's private contribution cost. The provider's utility is

\begin{equation}
\label{eq:up}
U_p
=
\pi V\mathbf{1}_{\left\{\sum_j e_j\geq X\right\}}
-p\sum_{i=1}^n e_i,
\end{equation}

where $\pi>0$ captures the provider's value from meeting the quality target relative to the benefit received by an individual user. The term $\pi V$ may reflect subscription revenue, API usage, or other benefits from successful adaptation, net of costs other than contribution subsidies. We assume that the provider pays $p$ for each unit of final effective contribution regardless of whether the aggregate target is ultimately reached. This contribution-contingent payment rule creates the possibility of subsidy leakage when positive contributions are collected but the target is not met.

Note that, while the model is stylized, it captures the two key features highlighted in the introduction: threshold-based AI improvement and data withdrawal. 
The qualitative insights extend to more general specifications, as long as key mathematical properties hold (e.g., convexity of costs).

\subsection{Contribution Game without Withdrawal}\label{subs}

Consider the benchmark mechanism where users cannot withdraw their data once submitted. This mechanism, denoted by \(C\), captures the standard subsidy-based approach to user data acquisition. A positive subsidy $p>0$ creates a private incentive for users to contribute data even when successful provision is uncertain. Users also value successful provision because meeting the predefined quality target generates the common benefit $V$. Therefore, mechanism \(C\) can lead to two qualitatively different equilibrium outcomes: a floor equilibrium, in which users contribute only for the subsidy, and a cutoff equilibrium, in which low-cost users contribute more to help reach the threshold.

For a user with privacy cost \(c\), if the user expects the aggregate data contribution to be insufficient for model upgrading, the provision term in the utility becomes irrelevant. The user's private optimization problem then reduces to maximizing \(pe-ce^2/2\) over \(e\in[0,1]\). The corresponding subsidy-induced floor contribution is \(e_0(c)=\min\{p/c,1\}\). When \(p/c\le 1\), the floor payoff is \(p^2/(2c)\). Thus, we have a floor equilibrium which is an equilibrium in which every user contributes only this floor amount.

\begin{prop}[Floor Equilibrium]\label{prop:floor}If the subsidy satisfies \(1+(n-1)\min\{p/\underline c,1\}<X\), then the strategy \(e_0(c_i)=\min\{p/c_i,\allowbreak 1\}\) adopted by all users constitutes a Bayesian Nash equilibrium.
\end{prop}
\begin{proof}
Since \(\min\{p/\underline c,1\}\) is the maximal floor contribution, condition \(1+(n-1)\min\{p/\underline c,1\}<X\) implies that even a unilateral deviation to \(e_i=1\) cannot make provision successful. Hence the provision term is irrelevant, and each user only solves \(\max_{e_i\in[0,1]} pe_i-c_i e_i^2/2\), whose solution is \(e_0(c_i)=\min\{p/c_i,\allowbreak 1\}\). Thus the floor profile is a best response for every type. The details of proof are shown in the Appendix.
\end{proof}

\noindent\textbf{Subsidy Leakage.} Under the floor equilibrium, the platform pays a total subsidy of $p\sum_i e_0(c_i)$, but the aggregate contribution does not reach the required target. We refer to this outcome as subsidy leakage: the subsidy
induces positive individual contributions, yet the collective data supply
remains below the contribution threshold required for successful provision.

\begin{prop}[Cutoff Equilibrium]\label{prop:Cutoff}
Assume \(p<\underline c\), so that \(e_0(c)=p/c\). Under the
binomial-concentration approximation, if the upgrading value \(V\) exceeds the
cutoff existence threshold \(\underline V_C(p,n,F)\), mechanism \(\mathcal C\)
supports a productive cutoff equilibrium. In this equilibrium, there exists a
cutoff \(a_p\in(\underline c,\bar c)\) such that low-cost users contribute above
the floor and high-cost users remain at the floor:
\begin{equation}\label{eq:C-strategy}
\begin{aligned}
e^{\mathcal C}(c)
&=
\begin{cases}
\tilde g(a_p), & c\le a_p,\\
p/c, & c>a_p,
\end{cases}\\
\tilde g(a_p)
&=\frac{X-(n-1-m)\mu_0(a_p)}{m+1},
\end{aligned}
\end{equation}
where \(\mu_0(a_p)=\mathbb E[p/c\mid c>a_p]\). The cutoff \(a_p\) is determined
by the marginal user's indifference condition:
\begin{equation}\label{eq:C-cutoff}
V\Big[B(\tilde g(a_p),a_p)-B(p/a_p,a_p)\Big]
=
\frac{(a_p\tilde g(a_p)-p)^2}{2a_p}.
\end{equation}
Here \(B(e_i,a)\) is the probability that aggregate contributions reach \(X\)
when user \(i\) contributes \(e_i\) and the other users follow cutoff \(a\).
Under the cutoff regularity condition, this productive cutoff is unique on the
provision-relevant branch.
\end{prop}

\begin{proof}
Define \(\Phi(a)=V\Delta B(a)-\Gamma^*(a,p,\tilde g(a))\), the net expected gain
from participating relative to the floor for a user with cost \(a\). When
\(V>\underline V_C(p,n,F)\), there exists a cutoff \(a_p\) such that
\(\Phi(a_p)=0\), which is exactly the marginal indifference condition in
\eqref{eq:C-cutoff}. Given this cutoff, the gain from participation for a type
\(c\) is
\[
\varphi(c)=V\Delta B(a_p)-\Gamma^*(c,p,\tilde g(a_p)).
\]
The pivotal term depends on the common cutoff strategy but not on the user's own
type, while \(\Gamma^*(c,p,\tilde g(a_p))\) is strictly increasing in \(c\) when
\(\tilde g(a_p)>p/c\). Hence \(\varphi(c)\) is strictly decreasing in \(c\).
Since \(\varphi(a_p)=0\), users with \(c<a_p\) strictly prefer the participation
contribution, and users with \(c>a_p\) strictly prefer the floor contribution.
Thus the strategy in \eqref{eq:C-strategy} is a symmetric Bayesian Nash
equilibrium.

Under the cutoff regularity condition, the ratio
\[R(a)=\Gamma^*(a,p,\tilde g(a))/\Delta B(a)\]
is strictly increasing on the
provision-relevant cutoff region. Therefore, the indifference equation
\(R(a)=V\) has at most one solution on that region, so the productive cutoff is
unique on the provision-relevant branch. The formal definition of
\(\underline V_C(p,n,F)\), the exact expression for \(B(e_i,a)\), the pivotal
probability, the binomial-concentration approximation, and the endpoint
best-response verification are provided in the Appendix.  
\end{proof}

The cutoff equilibrium does not eliminate the floor equilibrium of
mechanism \(\mathcal C\). Hence the same subsidy policy can support both a
non-provision floor outcome with subsidy leakage and a productive cutoff outcome.
Which equilibrium is reached depends on users' initial beliefs about others'
participation, so mechanism \(\mathcal C\) remains vulnerable to equilibrium
selection and coordination failure~\citep{schuch2025coordinating}.

\section{Withdrawal Mechanism}
\label{sec:withdrawal}

We now introduce a withdrawal mechanism with voluntary cost disclosure and personalized assignment. Users may reveal their privacy costs before training, the provider assigns contribution targets based on the revealed information, and users retain ex post withdrawal rights before the data are used.

\subsection{Motivation and Timing}

We consider that the provider can use disclosed costs to allocate excess contribution burdens to low-cost users while keeping high-cost users at their floor contributions. Ex post withdrawal rights reduce participation risk and encourage disclosure. We first assume disclosed costs are verifiable or auditable, so the strategic choice is whether to reveal rather than what cost to report; robustness to imperfect cost information is discussed in the Appendix. The mechanism proceeds in four stages.

\textbf{Stage 0 (Subsidy announcement and cost disclosure).} The provider pre-commits to a uniform per-unit subsidy $p \ge 0$. Each user $i$ independently decides whether to disclose his private cost parameter $c_i$ to the provider; let $R$ denote the set of revealers. The provider simultaneously makes the credible commitment that any user who does not disclose will not be required to supply any data, i.e., $g_i=0$, and will receive no subsidy.

\textbf{Stage 1 (Personalized assignment and data deletion).}
The provider observes the cost vector $\{c_i\}_{i\in R}$ of all revealers, and for each revealer $i$ assigns a notional data contribution $g_i \in [0,1]$, which is communicated to that user. Non-revealers are mandatorily assigned $g_i = 0$. Immediately after the assignment, the provider permanently deletes all raw cost information, retaining only the assignment profile $\{g_i\}$.

\textbf{Stage 2 (Data withdrawal).}
Each revealer $i$ observes his own assignment $g_i$ and the public aggregate notional contribution $G_0 = \sum_j g_j$, and then chooses a withdrawal amount $r_i \in [0, g_i]$. The final effective data contribution is $e_i = g_i-r_i$. We consider two implementable protocols:
\begin{itemize}
    \item \textbf{Simultaneous withdrawal ($\mathcal{S}$):} All revealers choose their withdrawals simultaneously, without observing each other's decisions, corresponding to centralized batch deletion.
    \item \textbf{Small-first withdrawal ($\mathcal{M}$):} Revealers act in increasing order of their notional assignments $g_i$ (with ties broken by decreasing cost), and each later mover observes the running aggregate of effective contributions after preceding withdrawals, corresponding to step-by-step deletion confirmation by the platform.
\end{itemize}

\textbf{Stage 3 (Model training and payoff settlement).}
The LLM is trained on the final effective contributions $\{e_i\}$. If $\sum_i e_i \ge X$, the model quality improves and \emph{all} users, including non-revealers, receive the benefit $V$; otherwise the benefit is $0$. User $i$'s expected payoff function remains the same as \eqref{eq:ui}, and the platform's expected profit function is given by \eqref{eq:up}.

\subsection{Equilibrium Characterization}
We now characterize the equilibrium of the withdrawal mechanism. First, there are two foundational properties holding under the design introduced, independent of the specific withdrawal protocol.

\begin{prop}[Full revelation under exclusion commitment]\label{prop:revelation} 
Suppose disclosed costs are verifiable or auditable. Under the provider's commitment to exclude non-revealers from assignment and subsidy payment, the strategy profile in which all users reveal their costs is a BNE for both withdrawal protocols. Revelation is weakly profitable for every user and strictly profitable whenever $p>0$ and a positive assignment is implemented with positive probability.
\end{prop}

\begin{proof}
Fix user \(i\) and suppose all other users reveal. Non-revelation implies \(g_i=e_i=0\), so \(i\) only obtains the public-good benefit with probability \(\Pr_{-i}\). Revelation weakly expands the provider's feasible assignment set, since the provider can always assign \(g_i=0\); hence the provision probability \(\Pr\) satisfies \(\Pr\geq \Pr_{-i}\). Moreover, whenever a positive assignment is implemented, the user can retain the floor contribution \[e^0(c_i)=\min\{p/c_i,1\},\]
which yields nonnegative subsidy surplus and is strictly profitable when \(p>0\). Therefore revelation is weakly profitable, and strictly profitable whenever positive assignment occurs with positive probability.
\end{proof}

\begin{prop}[Rational Assignment]\label{prop:Assignment}
In any subgame-perfect equilibrium (SPE) of the withdrawal-with-revelation game, the provider assigns a positive contribution profile only if the assignment is both provision-sustaining and provider-profitable. Specifically:
\begin{itemize}
    \item If the provider assigns $G_0\equiv\sum_j g_j>0$, then $\sum_j e_j \geq X$ in the subsequent Stage~2 equilibrium, and the provider's payoff satisfies $\Pi = \pi V - p\sum_j e_j \geq 0$.
    \item If the cost realization $(c_1,\ldots,c_n)$ is such that no assignment can lead to provision with nonnegative provider payoff in the Stage~2 equilibrium, the provider assigns $g_i=0$ for all $i$.
\end{itemize}
\end{prop}

\begin{proof}
Suppose the provider assigns $G_0>0$ but $\sum_j e_j < X$ in the Stage~2 equilibrium. Then $\Pi = -p\sum_j e_j < 0$. The provider could deviate to $g_i = 0$ for all $i$, yielding $\Pi = 0$---a profitable deviation, contradicting SPE. If no assignment can lead to provision, then every positive assignment yields $\Pi \leq 0$. The null assignment yields $\Pi = 0$, which is weakly optimal.
\end{proof}

Propositions~\ref{prop:revelation} and~\ref{prop:Assignment} imply that the withdrawal mechanism induces participation and rules out subsidy leakage after positive assignment. Moreover, a positive assignment serves as a coordination signal, selecting the provision outcome rather than a floor-like non-provision continuation. After observing $G_0>0$, users can infer that the provider, who has observed the realized cost vector, expects successful provision in the continuation equilibrium; otherwise, the provider would choose the null assignment.

\noindent\textbf{Assignment structure.} Conditional on revelation, the provider's optimal positive assignment takes the following form. Order users by their costs: $c_{(1)}\le \cdots\le c_{(n)}$. Let $e^0(c)= \min\{p/c,\,1\}$ be the floor contribution. The provider selects a backstop pool $K$ consisting of the $k$ lowest-cost users. Non-backstoppers are assigned their floor contributions, $g_{(j)}=e^0(c_{(j)})$ for $j\notin K$, while backstoppers are assigned protocol-specific targets, $g_{(j)}=e^*_{(j)}$ for $j\in K$, where
\begin{equation}\label{eq:dk}
\begin{aligned}
\sum_{j\in K} e^*_{(j)}
&=D_K\\
&\equiv X - \sum_{j\notin K} e^0(c_{(j)}),\\
e^*_{(j)}
&\in [e^0(c_{(j)}),1].
\end{aligned}
\end{equation}
Therefore, any positive assignment satisfies the following requirements.
\begin{equation}\label{eq:gk}
\begin{aligned}
G_0
&=\sum_j g_j\\
&=\sum_{j\notin K}e^0(c_{(j)})
  +\sum_{j\in K}e^*_{(j)}
=X.
\end{aligned}
\end{equation}
If no pool size $k$ satisfies the capacity condition $D_K\le k$ and the protocol-specific participation constraints, the provider chooses the null assignment $g_i=0$ for all users. Otherwise, conditional on a positive assignment, each non-backstopper has a dominant strategy in Stage~2 to retain exactly $e_j=g_j$, while a backstopper with cost $c$ who must retain $d\in[e^0(c),1]$ is willing to do so only if $V \ge \Gamma^*(c,p,d)$.

\begin{proof}
By Proposition~\ref{prop:floor}, every user voluntarily retains at least $e^0(c_i)$, so non-backstoppers need only be assigned their floor contributions. The remaining demand is placed on the lowest-cost users, and its exact split depends on the withdrawal protocol. A backstopper accepts his target when the quality gain covers the incremental privacy cost, i.e., $V\ge\Gamma^*$. The Appendix provides the formal derivation.
\end{proof}

After a positive assignment, users update their beliefs. If $G_0=0$, the game enters the null regime. If $G_0=X$, users infer that the cost vector lies in the protocol-specific provision region. Users with $g_i>e^0(c_i)$ identify themselves as backstoppers. Since the assignment exactly reaches $X$, any backstopper who withdraws below his target is pivotal and causes provision to fail.

\subsection{Single-Backstopper Case}
\label{Single-Backstopper Case}
When $D\le 1$, a single backstopper suffices.
Non-backstoppers have a dominant strategy to retain their floor $e_j=p/c_j$, since this is exactly their privately optimal amount.
The backstopper, observing the aggregate $G_0$, infers the deterministic residual gap $D$ and faces a single-agent binary choice: retain $D$ if $V\ge\Gamma^*(c_{(1)},p,D)$.
This decision depends only on the gap size, not on whether withdrawal is simultaneous or sequential.
Hence $\mathcal{S}$ and $\mathcal{M}$ yield identical outcomes, with provision region:
\begin{equation}\label{eq:regin1}
\begin{aligned}
\Omega_P^{\mathcal{S}}(p)
&=\Omega_P^{\mathcal{M}}(p)\\
&=\bigl\{(c_1,\dots,c_n):
D\le 1 \text{ and}\\
&\hspace{5.6em}
V\ge\Gamma^*(c_{(1)},p,D)\bigr\}.
\end{aligned}
\end{equation}
The result shows protocol design is irrelevant when only one backstopper is needed. We now show that meaningful protocol differences emerge only when multiple backstoppers are required.

\subsection{Multi-Backstopper Case}

When $D_K>1$, the provider must recruit $k\ge 2$ backstoppers. The two protocols produce different equilibrium structures. We now proceed to a comparative discussion of the two mechanisms.

\noindent\textbf{Small-First Withdrawal.}
Under $\mathcal{M}$ with backstop pool $K$ of size $k$, the withdrawal subgame has a unique PBE.
The equilibrium is constructed by backward induction. For the last backstopper $(1)$, define the success indicator
\begin{equation}\label{eq:S1}
\begin{aligned}
S_1(c_{(1)},E_{<1})
&=\mathbf{1}\{X-E_{<1}\le 1\}\\
&\quad\cdot\mathbf{1}\{V\ge
\Gamma^*(c_{(1)},p,X-E_{<1})\}.
\end{aligned}
\end{equation}
If $S_1=1$, he fills the residual gap $D_1=X-E_{<1}$; otherwise he withdraws to his floor $p/c_{(1)}$.
For $j=2,\dots,k$, backstopper $(j)$ observes $E_{<j}$ and chooses $e_{(j)} \in [p/c_{(j)}, 1]$ to minimize $c_{(j)}e_{(j)}^2/2$ subject to leaving a residual $X - E_{<j} - e_{(j)}$ that makes the subsequent chain succeed (i.e., $S_{j-1}=1$ under equilibrium continuation).
The provider assigns positively iff the chain succeeds from the start, i.e., $S_k=1$.
The resulting provision region is
\begin{equation}\label{eq:regin2}
\Omega_P^{\mathcal{M}}(p) = \bigl\{(c_1,\dots,c_n): S_k(c_{(k)},\dots,c_{(1)},E_F)=1\bigr\},
\end{equation}
where $E_F=\sum_{j\notin K} p/c_j$ is the deterministic floor aggregate. Provision requires $V \ge \Gamma^*\bigl(c_{(j)}, p, e_{(j)}^*\bigr)$ for every backstopper along the equilibrium path, with the last backstopper's condition $V \ge \Gamma^*(c_{(1)}, p, D_1)$ being the final binding link.

\noindent\textbf{Simultaneous Withdrawal.}
Under $\mathcal{S}$, all backstoppers choose retentions simultaneously without observing each other's actions. To enable coordination, the provider privately communicates personalized target contributions based on the equal-marginal-cost allocation:
\begin{equation}\label{eq:S-best-alloc}
    e_j^* = \frac{D_K}{\bar{C}_K \cdot c_j}, \qquad \text{where } \bar{C}_K \equiv \sum_{\ell \in K} \frac{1}{c_\ell},
\end{equation}
This allocation equalizes \(c_j e_j^*\) across backstoppers, gives lower targets to higher-cost users, and minimizes the binding participation constraint. The provision BNE exists iff all backstoppers are willing to retain their targets, i.e., \(V\ge \max_{j\in K}\Gamma^*(c_j,p,e_j^*)\). In this equilibrium, each backstopper retains \(e_j^*\); any unilateral withdrawal to the floor \(p/c_j\) would make the user pivotal and destroy provision.
\begin{equation}\label{eq:SV}
\begin{aligned}
V &\ge \max_{j \in K}\;
\Gamma^*\!\left(c_j,p,e_j^*\right)\\
&\Longleftrightarrow
V\ge\bar{V}^S(c,p)\\
&\equiv
\dfrac{\left(D_K/\bar{C}_K-p\right)^2}{2c_{(1)}}.
\end{aligned}
\end{equation}
The provision region is: 
\begin{equation}\label{eq:S-region}
    \Omega_P^{\mathcal{S}}(p) = \{(c_1,\dots,c_n): V \ge \bar{V}^S(c,p)\}.
\end{equation}

Intuitively, when $k\ge 2$, the sequential mechanism $\mathcal{M}$ exploits the flexibility created by the withdrawal order, while the simultaneous mechanism $\mathcal{S}$ adopts the equal-marginal-cost allocation and therefore often requires a higher quality-improvement value $V$ to satisfy all backstoppers' participation constraints.

\begin{prop}[Protocol containment]\label{prop:Containment} For any subsidy $p$, $\Omega_P^{\mathcal{S}}(p) \subseteq \Omega_P^{\mathcal{M}}(p)$, with equality when $k=1$ and strict inclusion when $k\ge 2$ generically. 
\end{prop}
The complete construction and an illustrative cost realization showing strict inclusion are provided in the Appendix.

\section{Mechanism comparisons}
\label{sec:comparison}

We now compare the equilibrium outcomes and welfare implications of the three mechanisms
\(C\), \(\mathcal S\), \(\mathcal M\). Mechanism \(C\) admits both a floor
equilibrium and a productive cutoff equilibrium, so users may make positive subsidy-induced
contributions even when provision fails. In contrast, the withdrawal mechanisms
\(\mathcal S\) and \(\mathcal M\) replace floor failure with a null assignment: positive
assignment is made only when provision can be sustained in the withdrawal subgame. Hence
\(\mathcal S\) and \(\mathcal M\) select either null provision or successful provision, while
\(C\) may generate a positive-contribution but non-provision outcome.

\noindent\textbf{Social welfare.}
For welfare comparison, we use total surplus rather than provider profit. Since subsidy payments cancel as transfers between the provider and users, expected social welfare under mechanism \(J\in\{\mathcal C,\mathcal S,\mathcal M\}\) is
\begin{equation}\label{eq:EW}
\begin{aligned}
SW^J(V,p)
=\mathbb{E}_{\mathbf c}\Bigl[
&\underbrace{nV\cdot\mathbf 1_{\{G^J(\mathbf c;V,p)\ge X\}}}
_{\text{provision benefit}}\\
&-\underbrace{\sum_{i=1}^n
\frac{c_i(e_i^J(\mathbf c;V,p))^2}{2}}
_{\text{total privacy cost}}
\Bigr].
\end{aligned}
\end{equation}
Thus, welfare depends on two objects: the provision probability
\(\Pr^J(V,p)\equiv\Pr(G^J(\mathbf c;V,p)\ge X)\) and the total privacy cost induced by equilibrium contributions.

Comparison at any given $(V,p)$, Three facts follow directly from the equilibrium characterization (Sections~\ref{sec:withdrawal} and~\ref{sec:comparison}):
\begin{itemize}
  \item \textbf{No floor failure.}
        Under $\mathcal{S}$ and $\mathcal{M}$ the null regime is no contributions are collected and no privacy cost is incurred. Under the floor equilibrium, \(C\) may induce positive subsidized contributions without successful provision, resulting in subsidy leakage. Thus $\mathcal{S},\mathcal{M}$ strictly dominate \(C\) in welfare whenever a floor equilibrium would be played in \(C\).
  \item \textbf{Provision region nesting.}
        For every $(p,V)$, $\Omega_P^S(p)\subseteq\Omega_P^M(p)$. Consequently $\Pr^S(p)\le\Pr^M(p)$, with equality in the single-backstopper case (Propositions~\ref{prop:Containment}).
  \item \textbf{Efficiency vs.\ robustness.}
Where both $\mathcal{S}$ and $\mathcal{M}$ succeed, $\mathcal{S}$ is more cost-efficient because its equal-marginal-cost allocation minimizes $\sum_i c_i e_i^2/2$. In contrast, $\mathcal{M}$ is more robust: its backward-induction chain places the most capable backstopper last, allowing provision under more heterogeneous cost realizations.
\end{itemize}

\section{Numerical experiments}
\label{sec:experiments}
This section presents numerical simulations that illustrate the equilibrium properties of the three mechanisms. The main theoretical prediction is that $C$ is sensitive to equilibrium selection, whereas $\mathcal{S}$ and $\mathcal{M}$ remove the floor-failure outcome via rational assignment.

\paragraph{Experimental setup.}
 We set \(n=50\), \(X=10.5\), and normalize the maximum individual contribution to one. Each privacy cost is independently drawn from the uniform distribution \(U[1,5]\). We vary the common benefit from successful provision, $V\in[0,5]$, and the per-unit subsidy \(p\in[0,0.65]\). For each pair \((V,p)\), we conduct 5,000 independent simulations of cost vectors and compute the equilibrium outcomes associated with each mechanism. The provision success probability is reported as the fraction of simulations in which aggregate contribution reaches $X$, and expected social welfare is computed via \eqref{eq:EW}.

\paragraph{Experiment 1 (Equilibrium selection in mechanism \(C\).)} Unlike mechanisms $\mathcal{M}$ and $\mathcal{S}$, mechanism $\mathcal{C}$ admits two equilibria under given conditions, the floor equilibrium and the cutoff equilibrium, whose selection depends on users' initial beliefs about the convergence outcome, and therefore requires separate treatment\citep{carlsson2024one}. To capture equilibrium selection, we introduce a belief parameter \(b_0\), interpreted as the maximum share of cutoff participants that users can coordinate on. When \(b_0=0\), only the floor equilibrium is selected. When \(b_0>0\), we search over cutoff beliefs \(q\le b_0\), evaluate \(\Phi(a;V,p)\) at \(a=F^{-1}(q)\), and select the supported cutoff with the largest incentive margin if \(\Phi\ge0\); otherwise the mechanism remains at the floor. Details are provided in the Appendix.

\paragraph{Belief sensitivity in mechanism \(C\).}
The belief parameter \(b_0\) has a direct effect on equilibrium selection in mechanism \(C\). Figure~\ref{fig:c-belief} reports the provision success probability under three belief levels: \(b_0=0\), \(b_0=0.15\), and \(b_0=0.3\). As \(b_0\) increases, more coexistence regions select the productive cutoff equilibrium rather than the floor equilibrium. Thus, the success probability of \(C\) is highly sensitive to users' equilibrium-selection beliefs.

\begin{figure}[t]
    \centering
    \includegraphics[width=\linewidth]{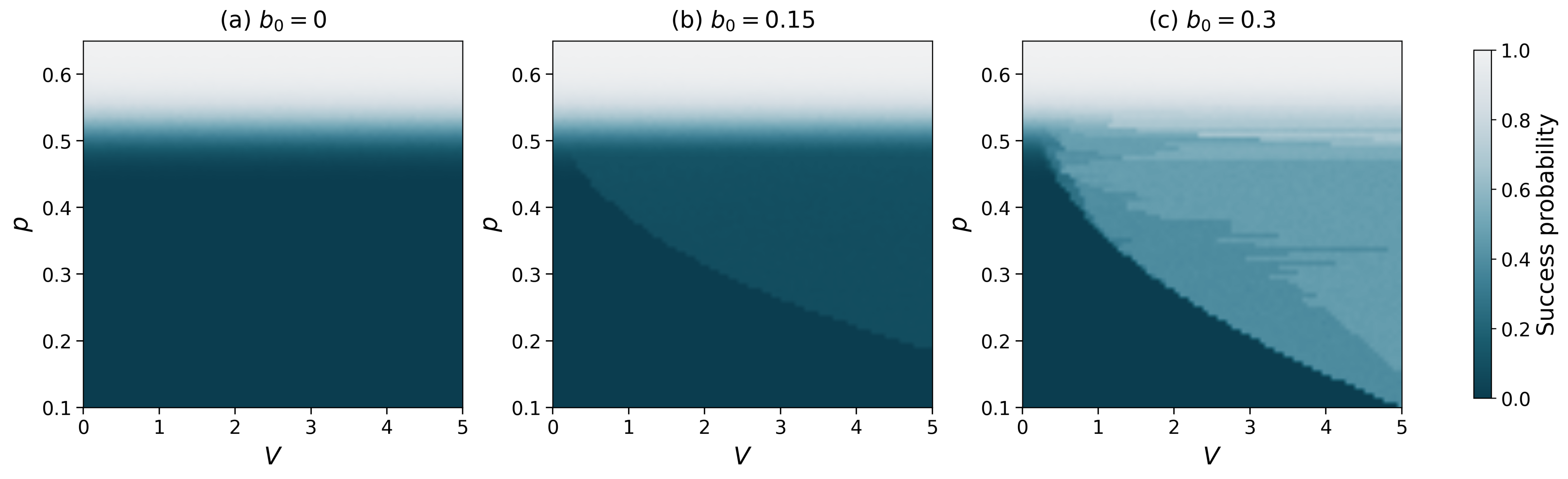}
    \caption{Provision success probability of mechanism \(C\) under different equilibrium-selection beliefs. When \(b_0=0\), only the floor equilibrium is selected, so success occurs only when floor contributions alone reach \(X\). Higher \(b_0\) allows productive cutoff equilibria to be selected in a larger region of the \((V,p)\) plane.}
    \label{fig:c-belief}
\Description{Three heatmaps compare provision probabilities for mechanism C across benefit and subsidy values at belief levels zero, 0.15 and 0.3.}
\end{figure}

\paragraph{Provision success across mechanisms.}We next compare the provision success probabilities of the three mechanisms \(C\), $\mathcal{S}$, and $\mathcal{M}$. For mechanism \(C\), we use the intermediate belief level \(b_0=0.15\), which represents a moderately optimistic coordination environment. In contrast, mechanisms $\mathcal{S}$ and $\mathcal{M}$ follow the assignment-and-withdrawal rules introduced above and do not require additional belief parameters.

Figure~\ref{fig:success-comparison} shows that the success probability for the same \((V,p)\) pair generally ranks as $\mathcal{M}$ > $\mathcal{S}$ > \(C\). For withdrawal mechanisms $\mathcal{S}$ and $\mathcal{M}$, the upgrading value \(V\) acts as a threshold: when \(V\) is below the mechanism-specific incremental participation cost, provision fails unless the subsidy \(p\) is large enough for floor contributions alone to reach the threshold. The figure also shows that the feasible region under \(M\) is substantially larger than that under $\mathcal{S}$).

\begin{figure}[t]
    \centering
    \includegraphics[width=\linewidth]{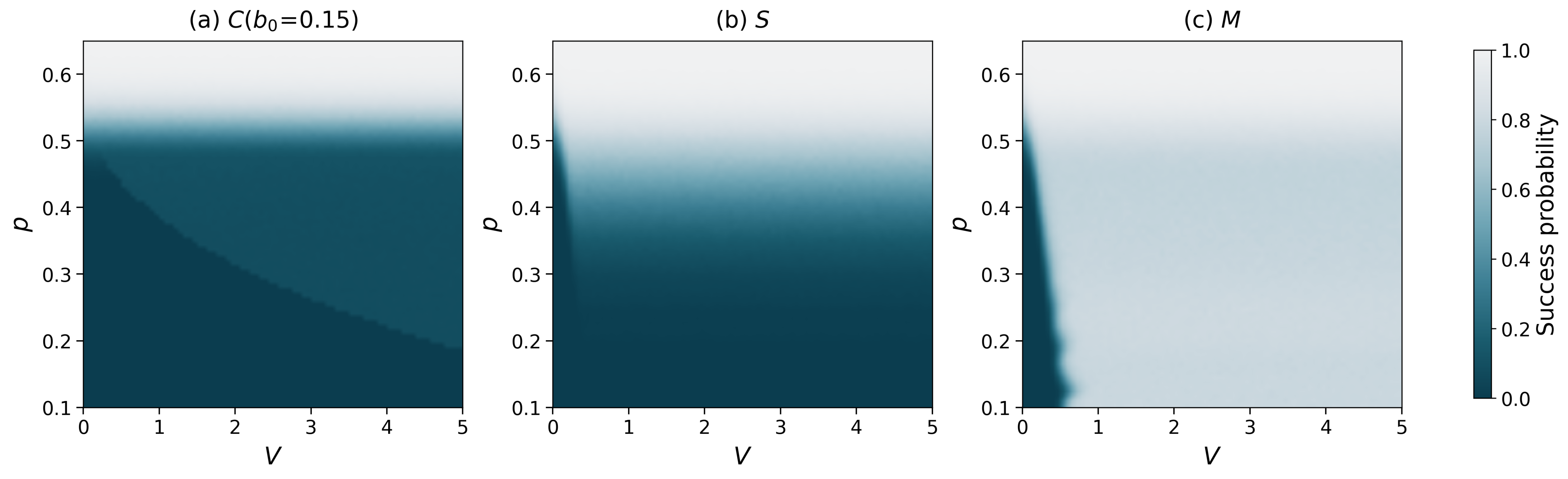}
    \caption{Provision success probability under the subsidy mechanism \(C\), simultaneous withdrawal $\mathcal{S}$, and small-first withdrawal $\mathcal{M}$. We report \(C\) at \(b_0=0.15\). The withdrawal mechanisms do not rely on an external optimistic belief to eliminate floor failure. The small-first protocol $\mathcal{M}$ has the largest success region.}
    \label{fig:success-comparison}
\Description{Three heatmaps compare provision probabilities for mechanisms C, S and M over benefit and subsidy values.}
\end{figure}

\paragraph{Experiment 2 (Social welfare comparison).}
The welfare comparison reinforces the provision-probability results. As shown in Figure~\ref{fig:welfare-comparison}, mechanism \(\mathcal C\) can generate negative total welfare when the floor equilibrium is selected: subsidies induce positive floor contributions and create provider-side subsidy leakage, while users still incur privacy costs without realizing the public improvement benefit. In contrast, under \(\mathcal S\) and \(\mathcal M\), failed provision leads to null assignment and no privacy cost. Negative welfare under the withdrawal mechanisms therefore arises only when provision succeeds but \(V\) is too small relative to the induced privacy costs. Since \(\mathcal M\) succeeds over a larger parameter region, it achieves higher expected welfare than \(\mathcal S\) in the reported simulations. This dominance is not a pointwise Pareto improvement, because backstoppers may bear higher privacy costs; when both \(\mathcal S\) and \(\mathcal M\) succeed, \(\mathcal S\) yields lower average privacy cost. Details are provided in the Appendix.

\begin{figure}[t]
    \centering
    \includegraphics[width=\linewidth]{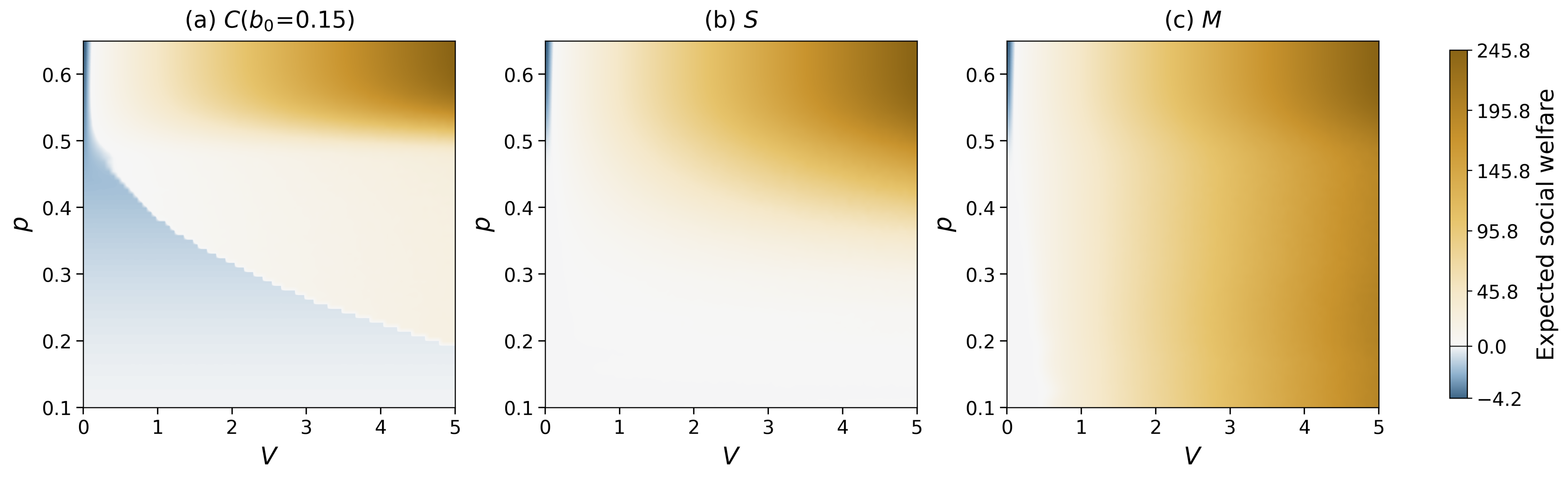}
    \caption{Expected social welfare under \(C\), $\mathcal{S}$, and $\mathcal{M}$. Social welfare is defined as the provision benefit $nV$ minus the total privacy cost. Regions with zero expected social welfare are set to white, while areas with negative welfare are marked in blue.}
    \label{fig:welfare-comparison}
\Description{Three heatmaps compare expected welfare for mechanisms C, S and M over benefit and subsidy values; blue denotes negative welfare and white denotes zero.}
\end{figure}

\subsection {LLM-Based Mechanism Comparison}
\label{sec:vten-experiment}
We conduct a dialogue-based experiment using LLM personas grounded in 50 anonymized profiles from the 2018 American Community Survey (ACS)~\citep{uscensus2018pums}, with ten profiles from each of five states. Each profile is evaluated under C, S, and M, yielding 150 paired trajectories. The experiment records contribution and withdrawal decisions during interactions with a target LLM. Candidate contributions undergo quality and utility review, and valid unsuccessful outcomes are retained.

We compare total contribution cost $K$, provision success $P$, and modeled social welfare $W$, where $I=\mathbf{1}\{G\geq102\}$, $W=6I-K$, and $P$ is the percentage of profiles achieving provision. Unlike the equilibrium simulations, dialogue trajectories can incur contribution costs without reaching the provision threshold.

\begin{table}[t]
\centering
\caption{Final outcomes across 50 matched profiles. Costs and welfare refer to the focal persona and two mathematical background participants at final settlement.}
\label{vten:tab:main}
\small
\begin{tabularx}{\linewidth}{@{}l*{3}{>{\centering\arraybackslash}X}@{}}
\toprule
Mechanism & Mean $K\downarrow$ & Mean $W\uparrow$ & Success $P\uparrow$ \\
\midrule
C & 2.2320 & 0.6480 & 48\% \\
S & 2.1338 & 1.8262 & 66\% \\
M & 2.2014 & 2.1186 & 72\% \\
\bottomrule
\end{tabularx}
\end{table}

Table~\ref{vten:tab:main} reports success rates of 48\%, 66\%, and 72\% for C, S, and M, respectively. M achieves the highest mean welfare, while S has the lowest unconditional mean cost. Figure~\ref{vten:fig:mechanism-heatmaps} presents individual outcomes, and Figure~\ref{vten:fig:five-person-success} shows cumulative success rates.

These results support the cost-efficiency--provision trade-off in the evaluated dialogue setting: simultaneous withdrawal incurs lower costs when both protocols succeed, whereas small-first withdrawal achieves higher provision success and mean welfare. Appendix~\ref{sec:vten-appendix} provides individual results and descriptive comparisons across education, state, and occupational groups.

\begin{figure}[!t]
\centering
\includegraphics[width=\linewidth]{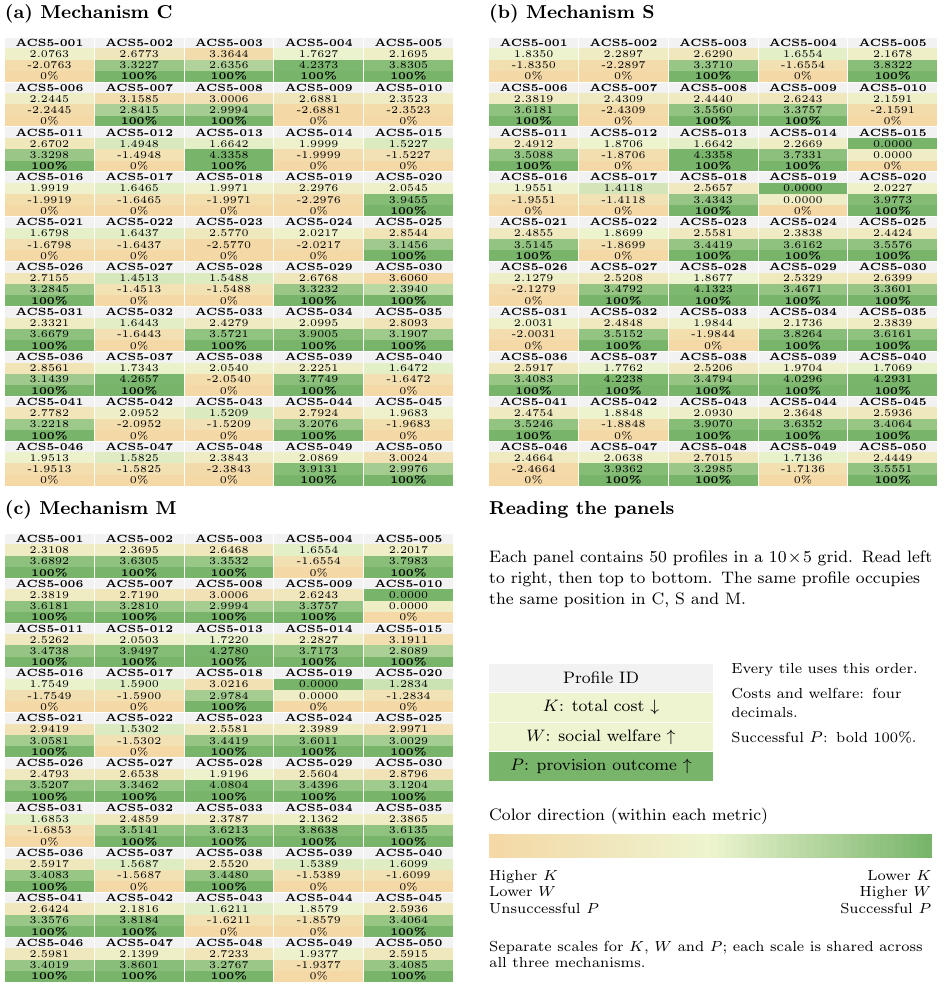}
\caption{Individual final outcomes from Table~\ref{vten:tab:persons}. Panels (a)--(c) show mechanisms C, S and M as square profile grids; the lower-right key explains the common encoding. Each tile reports profile ID, $K$, $W$ and $P$. All 450 values and all valid failures are retained; low cost alone does not imply successful provision.}
\label{vten:fig:mechanism-heatmaps}
\Description{Profile grids for C, S and M encode cost, welfare and provision for the same fifty profiles. The complete numerical values are tabulated in the appendix.}
\end{figure}
\begin{figure}[!t]
\centering
\includegraphics[width=.90\linewidth]{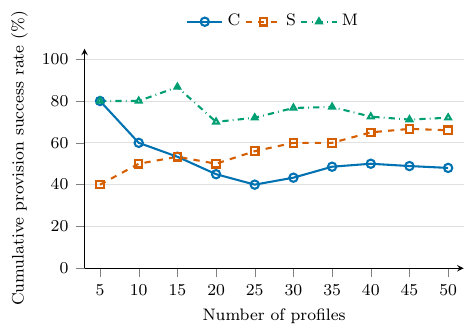}
\caption{Cumulative provision success rates for the first 5, 10, 15, \ldots, 50 profiles under C, S and M, calculated from Table~\ref{vten:tab:persons}. Each point uses the first $n$ profiles, including failures; every prefix has equal state representation. Points share observations and do not represent independent repeated runs.}
\label{vten:fig:five-person-success}
\Description{Three lines show cumulative provision success for C, S and M at five-profile increments, ending at 48, 66 and 72 percent respectively.}
\end{figure}

\section{Discussion}
\label{sec:discussion}

Our results offer practical guidance for LLM providers, users, and policymakers. For providers, rewarding participation or data volume is insufficient when model improvement requires crossing a quality cutoff: subsidies may generate costly data contributions without producing a usable improvement. Cost-informed assignments can instead coordinate users on a feasible target, while the withdrawal stage acts as a commitment device because withdrawing reduces both the user's payment and the common benefit from LLM quality. 
For users, the mechanism preserves control over data contributions while making the collective value of retaining them transparent, although safeguards may be needed to prevent low-cost backstoppers from bearing excessive privacy costs. For governments, the findings suggest that withdrawal rights need not conflict with AI innovation; regulation can encourage providers to disclose contribution targets, condition payments on meaningful model improvement, and protect heavily relied-upon contributors. These implications are qualitative, as the model abstracts from uncertain thresholds, heterogeneous data quality, imperfect cost verification, and the legal or technical difficulty of removing data influence from trained models.

\FloatBarrier
\bibliographystyle{unsrtnat}
\bibliography{references}

@article{ouyang2022training,
  title   = {Training language models to follow instructions with human feedback},
  author  = {Ouyang, Long and Wu, Jeff and Jiang, Xu and Almeida, Diogo and Wainwright, Carroll L. and Mishkin, Pamela and Zhang, Chong and Agarwal, Sandhini and Slama, Katarina and Ray, Alex and others},
  journal = {Advances in Neural Information Processing Systems},
  volume  = {35},
  pages   = {27730--27744},
  year    = {2022}
}

@article{hoffmann2022training,
  title   = {Training compute-optimal large language models},
  author  = {Hoffmann, Jordan and Borgeaud, Sebastian and Mensch, Arthur and Buchatskaya, Elena and Cai, Trevor and Rutherford, Eliza and de Las Casas, Diego and Hendricks, Lisa Anne and Welbl, Johannes and Clark, Aidan and others},
  journal = {arXiv preprint arXiv:2203.15556},
  year    = {2022}
}

@inproceedings{villalobos2024position,
  title     = {Position: Will we run out of data? Limits of LLM scaling based on human-generated data},
  author    = {Villalobos, Pablo and Ho, Anson and Sevilla, Jaime and Besiroglu, Tamay and Heim, Lennart and Hobbhahn, Marius},
  booktitle = {Forty-first International Conference on Machine Learning},
  year      = {2024}
}

@inproceedings{luo2025survey,
  title     = {A survey on efficient large language model training: From data-centric perspectives},
  author    = {Luo, Jie and Wu, Bing and Luo, Xiaolin and others},
  booktitle = {Proceedings of the 63rd Annual Meeting of the Association for Computational Linguistics (Volume 1: Long Papers)},
  pages     = {30904--30920},
  year      = {2025}
}

@inproceedings{sim2020collaborative,
  title     = {Collaborative machine learning with incentive-aware model rewards},
  author    = {Sim, R. H. L. and Zhang, Y. and Chan, M. C. and Low, B. K. H.},
  booktitle = {International Conference on Machine Learning},
  publisher = {PMLR},
  pages     = {8927--8936},
  year      = {2020}
}

@article{sim2023incentives,
  title   = {Incentives in private collaborative machine learning},
  author  = {Sim, Rachael Hwee Ling and Zhang, Yehong and Hoang, N. and Low, Bryan Kian Hsiang},
  journal = {Advances in Neural Information Processing Systems},
  volume  = {36},
  pages   = {7555--7593},
  year    = {2023}
}

@article{bagnoli1991voluntary,
  title   = {Voluntary contribution games: Efficient private provision of public goods},
  author  = {Bagnoli, Mark and McKee, Michael},
  journal = {Economic Inquiry},
  volume  = {29},
  number  = {2},
  pages   = {351--366},
  year    = {1991}
}

@article{bag2011sequential,
  title   = {On sequential and simultaneous contributions under incomplete information},
  author  = {Bag, Parimal Kanti and Roy, Santanu},
  journal = {International Journal of Game Theory},
  volume  = {40},
  number  = {1},
  pages   = {119--145},
  year    = {2011}
}

@article{compte2003voluntary,
  title   = {Voluntary contributions to a joint project with asymmetric agents},
  author  = {Compte, Olivier and Jehiel, Philippe},
  journal = {Journal of Economic Theory},
  volume  = {112},
  number  = {2},
  pages   = {334--342},
  year    = {2003}
}

@article{marx2000dynamic,
  title   = {Dynamic voluntary contribution to a public project},
  author  = {Marx, Leslie M. and Matthews, Steven A.},
  journal = {The Review of Economic Studies},
  volume  = {67},
  number  = {2},
  pages   = {327--358},
  year    = {2000}
}

@article{varian1994sequential,
  title   = {Sequential contributions to public goods},
  author  = {Varian, Hal R.},
  journal = {Journal of Public Economics},
  volume  = {53},
  number  = {2},
  pages   = {165--186},
  year    = {1994}
}

@article{christiano2017deep,
  title   = {Deep reinforcement learning from human preferences},
  author  = {Christiano, Paul F. and Leike, Jan and Brown, Tom and Martic, Miljan and Legg, Shane and Amodei, Dario},
  journal = {Advances in Neural Information Processing Systems},
  volume  = {30},
  year    = {2017}
}

@article{stiennon2020learning,
  title   = {Learning to summarize with human feedback},
  author  = {Stiennon, Nisan and Ouyang, Long and Wu, Jeff and Ziegler, Daniel M. and Lowe, Ryan and Voss, Chelsea and Radford, Alec and Amodei, Dario and Christiano, Paul F.},
  journal = {Advances in Neural Information Processing Systems},
  volume  = {33},
  pages   = {3008--3021},
  year    = {2020}
}

@article{acquisti2016economics,
  title   = {The economics of privacy},
  author  = {Acquisti, Alessandro and Taylor, Curtis and Wagman, Liad},
  journal = {Journal of Economic Literature},
  volume  = {54},
  number  = {2},
  pages   = {442--492},
  year    = {2016}
}

@inproceedings{carlini2021extracting,
  title     = {Extracting training data from large language models},
  author    = {Carlini, Nicholas and Tramer, Florian and Wallace, Eric and Jagielski, Matthew and Herbert-Voss, Ariel and Lee, Katherine and Roberts, Adam and Brown, Tom B. and Song, Dawn and Erlingsson, {\'U}lfar and others},
  booktitle = {30th USENIX Security Symposium (USENIX Security 21)},
  pages     = {2633--2650},
  year      = {2021}
}

@article{vanhuyck1990tacit,
  title   = {Tacit coordination games, strategic uncertainty, and coordination failure},
  author  = {Van Huyck, John B. and Battalio, Raymond C. and Beil, Richard O.},
  journal = {The American Economic Review},
  volume  = {80},
  number  = {1},
  pages   = {234--248},
  year    = {1990}
}

@article{zhan2021survey,
  title   = {A survey of incentive mechanism design for federated learning},
  author  = {Zhan, Yufeng and Zhang, Jie and Hong, Zicong and Wu, Leijie and Li, Peng and Guo, Song},
  journal = {IEEE Transactions on Emerging Topics in Computing},
  volume  = {10},
  number  = {2},
  pages   = {1035--1044},
  year    = {2021}
}

@article{zhang2024datamarkets,
  title   = {A survey on data markets},
  author  = {Zhang, Jun and Bi, Yuxin and Cheng, Meng and Yang, Qiang},
  journal = {arXiv preprint arXiv:2411.07267},
  year    = {2024}
}

@article{ginart2019making,
  title   = {Making AI forget you: Data deletion in machine learning},
  author  = {Ginart, Antonio A. and Guan, Melody Y. and Valiant, Gregory and Zou, James},
  journal = {Advances in Neural Information Processing Systems},
  volume  = {32},
  pages   = {3518--3531},
  year    = {2019}
}

@article{guo2019certified,
  title   = {Certified data removal from machine learning models},
  author  = {Guo, C. and Goldstein, T. and Hannun, A. and et al.},
  journal = {arXiv preprint arXiv:1911.03030},
  year    = {2019}
}

@article{palfrey1984participation,
  title   = {Participation and the provision of discrete public goods: a strategic analysis},
  author  = {Palfrey, T. R. and Rosenthal, H.},
  journal = {Journal of Public Economics},
  volume  = {24},
  number  = {2},
  pages   = {171--193},
  year    = {1984}
}

@article{croson2000step,
  title   = {Step returns in threshold public goods: A meta-and experimental analysis},
  author  = {Croson, R. T. A. and Marks, M. B.},
  journal = {Experimental Economics},
  volume  = {2},
  number  = {3},
  pages   = {239--259},
  year    = {2000}
}

@article{fan2025data,
  title   = {Do Data Valuations Make Good Data Prices?},
  author  = {Fan, Dongyang and Rotello, Tyler J. and Karimireddy, Sai Praneeth},
  journal = {arXiv preprint arXiv:2504.05563},
  year    = {2025}
}

@article{carlsson2024one,
  title   = {One bad apple spoils the barrel? Public good provision under threshold uncertainty},
  author  = {Carlsson, F. and Ek, C. and Lange, A.},
  journal = {Experimental Economics},
  volume  = {27},
  number  = {3},
  pages   = {664--686},
  year    = {2024}
}

@article{schuch2025coordinating,
  title   = {Coordinating on good and bad outcomes in threshold games--Evidence from an artefactual field experiment in Cambodia},
  author  = {Schuch, Esther and Nhim, Tum and Richter, Andries},
  journal = {Ecological Economics},
  volume  = {232},
  pages   = {108547},
  year    = {2025}
}

@article{wei2022emergent,
  title   = {Emergent Abilities of Large Language Models},
  author  = {Jason Wei and Yi Tay and Rishi Bommasani and Colin Raffel and Barret Zoph and Sebastian Borgeaud and Dani Yogatama and Maarten Bosma and Denny Zhou and Donald Metzler and Ed H. Chi and Tatsunori Hashimoto and Oriol Vinyals and Percy Liang and Jeff Dean and William Fedus},
  journal = {Transactions on Machine Learning Research},
  issn    = {2835-8856},
  year    = {2022},
  url     = {https://openreview.net/forum?id=yzkSU5zdwD},
  note    = {Survey Certification}
}

@article{grattafiori2024llama,
  title   = {The llama 3 herd of models},
  author  = {Grattafiori, Aaron and Dubey, Abhimanyu and Jauhri, Abhinav and Pandey, Abhinav and Kadian, Abhishek and Al-Dahle, Ahmad and Letman, Aiesha and Mathur, Akhil and Schelten, Alan and Vaughan, Alex and others},
  journal = {arXiv preprint arXiv:2407.21783},
  year    = {2024}
}

@article{acemoglu2022too,
  title     = {Too much data: Prices and inefficiencies in data markets},
  author    = {Acemoglu, Daron and Makhdoumi, Ali and Malekian, Azarakhsh and Ozdaglar, Asu},
  journal   = {American Economic Journal: Microeconomics},
  volume    = {14},
  number    = {4},
  pages     = {218--256},
  year      = {2022},
  publisher = {American Economic Association 2014 Broadway, Suite 305, Nashville, TN 37203-2425}
}

@article{sandler2015collective,
  title     = {Collective action: fifty years later},
  author    = {Sandler, Todd},
  journal   = {Public Choice},
  volume    = {164},
  pages     = {195--216},
  year      = {2015},
  publisher = {Springer}
}

@misc{gdpr2016,
  author       = {{European Parliament and Council of the European Union}},
  title        = {Regulation ({EU}) 2016/679 of the European Parliament and of the Council (General Data Protection Regulation)},
  year         = {2016},
  howpublished = {Official Journal of the European Union, L 119, 1--88},
  url          = {https://eur-lex.europa.eu/eli/reg/2016/679/oj}
}

@inproceedings{duetting2024mechanism,
  author    = {D{\"u}tting, Paul and Mirrokni, Vahab and
               Paes Leme, Renato and Xu, Haifeng and Zuo, Song},
  title     = {Mechanism Design for Large Language Models},
  booktitle = {Proceedings of the ACM Web Conference 2024},
  pages     = {144--155},
  year      = {2024},
  publisher = {Association for Computing Machinery},
  doi       = {10.1145/3589334.3645511},
  url       = {https://doi.org/10.1145/3589334.3645511}
}

@inproceedings{bergemann2025llm,
  author    = {Bergemann, Dirk and Bonatti, Alessandro and Smolin, Alex},
  title     = {The Economics of Large Language Models: Token Allocation,
               Fine-Tuning, and Optimal Pricing},
  booktitle = {Proceedings of the 26th ACM Conference on Economics
               and Computation},
  pages     = {786},
  year      = {2025},
  publisher = {Association for Computing Machinery},
  doi       = {10.1145/3736252.3742625},
  url       = {https://doi.org/10.1145/3736252.3742625}
}

@inproceedings{chen2022selling,
  author    = {Chen, Junjie and Li, Minming and Xu, Haifeng},
  title     = {Selling Data to a Machine Learner: Pricing via Costly Signaling},
  booktitle = {Proceedings of the 39th International Conference on Machine Learning},
  series    = {Proceedings of Machine Learning Research},
  volume    = {162},
  pages     = {3336--3359},
  year      = {2022},
  publisher = {PMLR},
  url       = {https://proceedings.mlr.press/v162/chen22m.html}
}

@inproceedings{zhang2025fairshare,
  author    = {Zhang, Luyang and Jiao, Cathy and Li, Beibei and Xiong, Chenyan},
  title     = {Fairshare Data Pricing via Data Valuation for Large Language Models},
  booktitle = {Advances in Neural Information Processing Systems},
  volume    = {38},
  year      = {2025},
  publisher = {Curran Associates, Inc.},
  url       = {https://proceedings.neurips.cc/paper_files/paper/2025/hash/2fed719c7af739e5a8615170629ed7fc-Abstract-Conference.html}
}

@article{wang2025paid,
  author  = {Wang, Bingchen and Wu, Zhaoxuan and Liu, Fusheng and
             Low, Bryan Kian Hsiang},
  title   = {Paid with Models: Optimal Contract Design for Collaborative
             Machine Learning},
  journal = {Proceedings of the AAAI Conference on Artificial Intelligence},
  volume  = {39},
  number  = {13},
  pages   = {14176--14183},
  year    = {2025},
  doi     = {10.1609/aaai.v39i13.33552},
  url     = {https://doi.org/10.1609/aaai.v39i13.33552}
}

@misc{uscensus2018pums,
  author = {{U.S. Census Bureau}},
  title = {2018 {American Community Survey} {Public Use Microdata Sample} ({PUMS})},
  year = {2018},
  url = {https://www.census.gov/programs-surveys/acs/microdata/access.2018.html},
  note = {Data set}
}

@misc{chatterji2025how,
  author = {Aaron Chatterji and Tom Cunningham and David Deming and Zo{\"e} Hitzig and Christopher Ong and Carl Shan and Kevin Wadman},
  title = {How People Use {ChatGPT}},
  year = {2025},
  note = {Working paper},
  url = {https://cdn.openai.com/pdf/a253471f-8260-40c6-a2cc-aa93fe9f142e/economic-research-chatgpt-usage-paper.pdf}
}

\clearpage
\appendix
\section{Dialogue experiment: detailed results}\label{sec:vten-appendix}
\subsection{Individual outcomes and profile characteristics}
Tables~\ref{vten:tab:persons} and~\ref{vten:tab:profile} report all 50 matched profiles. $K,W,P$ follow Table~\ref{vten:tab:main}; individual $P$ is a realized binary outcome, not a repeated-run probability. Age is the age in the 2018 survey. State, education and occupation are source-record attributes, while the cost coefficient and fictional tasks are assumptions. CA, FL, NY, TX and WA denote California, Florida, New York, Texas and Washington. College 1y+ denotes at least one year of college credit without a degree; GED includes alternative credentials. Topic codes are editorial identifiers for the dialogue topics.

\subsection{Subgroup comparisons}
Overall means are reported in Table~\ref{vten:tab:main}. Table~\ref{vten:tab:education} retains all valid outcomes in each original education category. $K,W$ are profile-weighted means and $P$ is success in percent. The figure pools categories below an associate's degree ($n=28$); it preserves the explicitly disclosed professional-group display exception.

Each state has ten paired profiles. Table~\ref{vten:tab:state} and Figure~\ref{vten:fig:nature-table06_state} compare cost, welfare and provision across the same profiles under all mechanisms.

The 38 detailed occupations are pooled into eight mutually exclusive editorial groups, not an official ACS classification. All 50 profiles and all valid failures remain included. Group means pool individual trajectories rather than averaging occupation-level percentages. Food-service, construction and agricultural managers are assigned to their respective domains. Retail and personal services includes cashiers, sales staff, barbers and childcare workers; agriculture and maintenance includes farming, grounds/tree work, cleaning and housekeeping. Detailed occupations appear in Table~\ref{vten:tab:profile}.

Education, occupation and state overlap in this small selected sample. Group rankings are descriptive, not separate causal effects; cells with one profile describe one case. Lower unconditional cost need not imply greater cost efficiency. Welfare obeys $W=6I-K$ by construction, so its threshold jump is not evidence of a universal empirical cost cutoff.

\begin{figure}[!htbp]
\centering
\includegraphics[width=.90\linewidth]{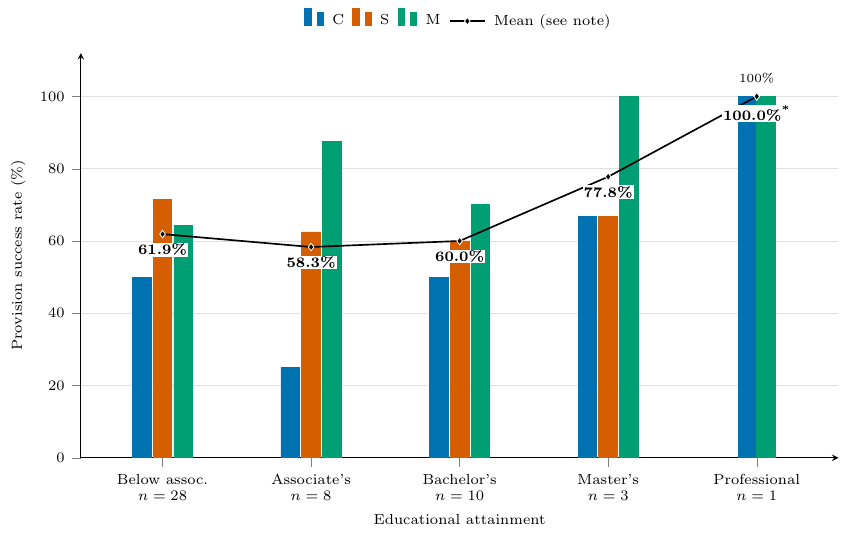}
\caption{Provision success rates by educational attainment under C, S and M. Profiles below an associate's degree are pooled into one group ($n=28$). The first four clusters compare all three mechanisms for the same profiles; the professional group displays C and M only. The black line shows the mean across C, S and M for the first four groups. The professional-group endpoint, marked with an asterisk, instead shows the C/M-only mean of 100\%, excluding the S failure. Its unadjusted C/S/M mean is 66.7\%; the endpoint uses a different denominator and is a post-hoc display choice, not a statistical exclusion.}
\label{vten:fig:education-scatter}
\Description{Grouped bars compare provision success across five education groups. The professional group shows only C and M; its unadjusted three-mechanism mean is disclosed in the caption.}
\end{figure}
\begingroup
\small
\setlength{\tabcolsep}{7pt}
\begin{longtable}{@{}r*{9}{r}@{}}
\caption{All 50 profiles: final group outcomes under C, S and M. IDs omit the ACS5- prefix; $P=1$ denotes success and $P=0$ failure. Means $(K,W,P)$ are C: $(2.2320,0.6480,48\%)$, S: $(2.1338,1.8262,66\%)$, M: $(2.2014,2.1186,72\%)$.}\label{vten:tab:persons}\\
\toprule
& \multicolumn{3}{c}{C} & \multicolumn{3}{c}{S} & \multicolumn{3}{c}{M} \\
\cmidrule(lr){2-4}\cmidrule(lr){5-7}\cmidrule(l){8-10}
ID & $K$ & $W$ & $P$ & $K$ & $W$ & $P$ & $K$ & $W$ & $P$\\
\midrule
\endfirsthead
\caption[]{All 50 profiles: final group outcomes (continued).}\\
\toprule
& \multicolumn{3}{c}{C} & \multicolumn{3}{c}{S} & \multicolumn{3}{c}{M} \\
\cmidrule(lr){2-4}\cmidrule(lr){5-7}\cmidrule(l){8-10}
ID & $K$ & $W$ & $P$ & $K$ & $W$ & $P$ & $K$ & $W$ & $P$\\
\midrule
\endhead
\bottomrule\endfoot
001 & 2.0763 & -2.0763 & 0 & 1.8350 & -1.8350 & 0 & 2.3108 & 3.6892 & 1 \\
002 & 2.6773 & 3.3227 & 1 & 2.2897 & -2.2897 & 0 & 2.3695 & 3.6305 & 1 \\
003 & 3.3644 & 2.6356 & 1 & 2.6290 & 3.3710 & 1 & 2.6468 & 3.3532 & 1 \\
004 & 1.7627 & 4.2373 & 1 & 1.6554 & -1.6554 & 0 & 1.6554 & -1.6554 & 0 \\
005 & 2.1695 & 3.8305 & 1 & 2.1678 & 3.8322 & 1 & 2.2017 & 3.7983 & 1 \\
006 & 2.2445 & -2.2445 & 0 & 2.3819 & 3.6181 & 1 & 2.3819 & 3.6181 & 1 \\
007 & 3.1585 & 2.8415 & 1 & 2.4309 & -2.4309 & 0 & 2.7190 & 3.2810 & 1 \\
008 & 3.0006 & 2.9994 & 1 & 2.4440 & 3.5560 & 1 & 3.0006 & 2.9994 & 1 \\
009 & 2.6881 & -2.6881 & 0 & 2.6243 & 3.3757 & 1 & 2.6243 & 3.3757 & 1 \\
010 & 2.3523 & -2.3523 & 0 & 2.1591 & -2.1591 & 0 & 0.0000 & 0.0000 & 0 \\
011 & 2.6702 & 3.3298 & 1 & 2.4912 & 3.5088 & 1 & 2.5262 & 3.4738 & 1 \\
012 & 1.4948 & -1.4948 & 0 & 1.8706 & -1.8706 & 0 & 2.0503 & 3.9497 & 1 \\
013 & 1.6642 & 4.3358 & 1 & 1.6642 & 4.3358 & 1 & 1.7220 & 4.2780 & 1 \\
014 & 1.9999 & -1.9999 & 0 & 2.2669 & 3.7331 & 1 & 2.2827 & 3.7173 & 1 \\
015 & 1.5227 & -1.5227 & 0 & 0.0000 & 0.0000 & 0 & 3.1911 & 2.8089 & 1 \\
016 & 1.9919 & -1.9919 & 0 & 1.9551 & -1.9551 & 0 & 1.7549 & -1.7549 & 0 \\
017 & 1.6465 & -1.6465 & 0 & 1.4118 & -1.4118 & 0 & 1.5900 & -1.5900 & 0 \\
018 & 1.9971 & -1.9971 & 0 & 2.5657 & 3.4343 & 1 & 3.0216 & 2.9784 & 1 \\
019 & 2.2976 & -2.2976 & 0 & 0.0000 & 0.0000 & 0 & 0.0000 & 0.0000 & 0 \\
020 & 2.0545 & 3.9455 & 1 & 2.0227 & 3.9773 & 1 & 1.2834 & -1.2834 & 0 \\
021 & 1.6798 & -1.6798 & 0 & 2.4855 & 3.5145 & 1 & 2.9419 & 3.0581 & 1 \\
022 & 1.6437 & -1.6437 & 0 & 1.8699 & -1.8699 & 0 & 1.5302 & -1.5302 & 0 \\
023 & 2.5770 & -2.5770 & 0 & 2.5581 & 3.4419 & 1 & 2.5581 & 3.4419 & 1 \\
024 & 2.0217 & -2.0217 & 0 & 2.3838 & 3.6162 & 1 & 2.3989 & 3.6011 & 1 \\
025 & 2.8544 & 3.1456 & 1 & 2.4424 & 3.5576 & 1 & 2.9971 & 3.0029 & 1 \\
026 & 2.7155 & 3.2845 & 1 & 2.1279 & -2.1279 & 0 & 2.4793 & 3.5207 & 1 \\
027 & 1.4513 & -1.4513 & 0 & 2.5208 & 3.4792 & 1 & 2.6538 & 3.3462 & 1 \\
028 & 1.5488 & -1.5488 & 0 & 1.8677 & 4.1323 & 1 & 1.9196 & 4.0804 & 1 \\
029 & 2.6768 & 3.3232 & 1 & 2.5329 & 3.4671 & 1 & 2.5604 & 3.4396 & 1 \\
030 & 3.6060 & 2.3940 & 1 & 2.6399 & 3.3601 & 1 & 2.8796 & 3.1204 & 1 \\
031 & 2.3321 & 3.6679 & 1 & 2.0031 & -2.0031 & 0 & 1.6853 & -1.6853 & 0 \\
032 & 1.6443 & -1.6443 & 0 & 2.4848 & 3.5152 & 1 & 2.4859 & 3.5141 & 1 \\
033 & 2.4279 & 3.5721 & 1 & 1.9844 & -1.9844 & 0 & 2.3787 & 3.6213 & 1 \\
034 & 2.0995 & 3.9005 & 1 & 2.1736 & 3.8264 & 1 & 2.1362 & 3.8638 & 1 \\
035 & 2.8093 & 3.1907 & 1 & 2.3839 & 3.6161 & 1 & 2.3865 & 3.6135 & 1 \\
036 & 2.8561 & 3.1439 & 1 & 2.5917 & 3.4083 & 1 & 2.5917 & 3.4083 & 1 \\
037 & 1.7343 & 4.2657 & 1 & 1.7762 & 4.2238 & 1 & 1.5687 & -1.5687 & 0 \\
038 & 2.0540 & -2.0540 & 0 & 2.5206 & 3.4794 & 1 & 2.5520 & 3.4480 & 1 \\
039 & 2.2251 & 3.7749 & 1 & 1.9704 & 4.0296 & 1 & 1.5389 & -1.5389 & 0 \\
040 & 1.6472 & -1.6472 & 0 & 1.7069 & 4.2931 & 1 & 1.6099 & -1.6099 & 0 \\
041 & 2.7782 & 3.2218 & 1 & 2.4754 & 3.5246 & 1 & 2.6424 & 3.3576 & 1 \\
042 & 2.0952 & -2.0952 & 0 & 1.8848 & -1.8848 & 0 & 2.1816 & 3.8184 & 1 \\
043 & 1.5209 & -1.5209 & 0 & 2.0930 & 3.9070 & 1 & 1.6211 & -1.6211 & 0 \\
044 & 2.7924 & 3.2076 & 1 & 2.3648 & 3.6352 & 1 & 1.8579 & -1.8579 & 0 \\
045 & 1.9683 & -1.9683 & 0 & 2.5936 & 3.4064 & 1 & 2.5936 & 3.4064 & 1 \\
046 & 1.9513 & -1.9513 & 0 & 2.4664 & -2.4664 & 0 & 2.5981 & 3.4019 & 1 \\
047 & 1.5825 & -1.5825 & 0 & 2.0638 & 3.9362 & 1 & 2.1399 & 3.8601 & 1 \\
048 & 2.3843 & -2.3843 & 0 & 2.7015 & 3.2985 & 1 & 2.7233 & 3.2767 & 1 \\
049 & 2.0869 & 3.9131 & 1 & 1.7136 & -1.7136 & 0 & 1.9377 & -1.9377 & 0 \\
050 & 3.0024 & 2.9976 & 1 & 2.4449 & 3.5551 & 1 & 2.5915 & 3.4085 & 1 \\
\end{longtable}
\endgroup

\clearpage
\begingroup
\small
\setlength{\tabcolsep}{4pt}
\renewcommand{\arraystretch}{1.10}
\begin{longtable}{@{}lrr>{\raggedright\arraybackslash}p{23mm}>{\raggedright\arraybackslash}p{63mm}rr@{}}
\caption{Survey-derived characteristics, assumed cost and post-hoc dialogue topic.}\label{vten:tab:profile}\\
\toprule Profile & Age & State & Education & Occupation & $c$ & Topic\\\midrule
\endfirsthead
\caption[]{Survey-derived characteristics (continued).}\\
\toprule Profile & Age & State & Education & Occupation & $c$ & Topic\\\midrule
\endhead
\bottomrule\endfoot

ACS5-001 & 51 & CA & Bachelor & Registered nurses & 2.9824 & T01 \\
ACS5-002 & 34 & FL & Associate & Computer programmers & 3.1720 & T02 \\
ACS5-003 & 63 & NY & Bachelor & General and operations managers & 4.5317 & T03 \\
ACS5-004 & 31 & TX & Bachelor & Food preparation workers & 1.2519 & T04 \\
ACS5-005 & 34 & WA & Bachelor & Accountants and auditors & 2.5520 & T05 \\
ACS5-006 & 35 & CA & Associate & Home health aides & 3.2351 & T01 \\
ACS5-007 & 52 & FL & College \textless{}1y & Construction managers & 3.7502 & T06 \\
ACS5-008 & 39 & NY & College 1y+ & Firefighters & 3.5820 & T07 \\
ACS5-009 & 19 & TX & HS diploma & Cashiers & 4.4826 & T08 \\
ACS5-010 & 65 & WA & Associate & Cooks & 4.2609 & T04 \\
ACS5-011 & 55 & CA & Associate & Office clerks, general & 3.8014 & T03 \\
ACS5-012 & 50 & FL & Associate & Janitors and building cleaners & 2.2086 & T10 \\
ACS5-013 & 21 & NY & HS diploma & Cooks & 1.0439 & T04 \\
ACS5-014 & 33 & TX & No schooling & Landscaping and groundskeeping workers & 2.8561 & T09 \\
ACS5-015 & 26 & WA & Bachelor & Web developers & 3.8154 & T02 \\
ACS5-016 & 58 & CA & Grade 9 & Carpenters & 2.2610 & T06 \\
ACS5-017 & 25 & FL & Bachelor & Customer service representatives & 1.2222 & T03 \\
ACS5-018 & 21 & NY & Bachelor & Computer support specialists & 4.1946 & T02 \\
ACS5-019 & 64 & TX & College \textless{}1y & Bookkeeping, accounting, and auditing clerks & 4.2914 & T05 \\
ACS5-020 & 34 & WA & Bachelor & Registered nurses & 1.5422 & T01 \\
ACS5-021 & 56 & CA & Master & Accountants and auditors & 3.7740 & T05 \\
ACS5-022 & 21 & FL & HS diploma & Construction laborers & 1.9739 & T06 \\
ACS5-023 & 27 & NY & Bachelor & Food service managers & 4.1281 & T04 \\
ACS5-024 & 21 & TX & HS diploma & Retail salespersons & 3.3225 & T08 \\
ACS5-025 & 62 & WA & College \textless{}1y & Farmers, ranchers, and other agricultural managers & 3.5748 & T09 \\
ACS5-026 & 37 & CA & Master & Bookkeeping, accounting, and auditing clerks & 3.7410 & T05 \\
ACS5-027 & 55 & FL & Associate & Nursing assistants & 3.9467 & T01 \\
ACS5-028 & 34 & NY & HS diploma & Maids and housekeeping cleaners & 1.2723 & T10 \\
ACS5-029 & 25 & TX & HS diploma & Tree trimmers and pruners & 4.0129 & T09 \\
ACS5-030 & 24 & WA & Bachelor & Data entry keyers & 4.6451 & T03 \\
ACS5-031 & 27 & CA & College 1y+ & Firefighters & 2.7394 & T07 \\
ACS5-032 & 19 & FL & College 1y+ & Web developers & 3.7707 & T02 \\
ACS5-033 & 49 & NY & Professional & Mechanical engineers & 2.4924 & T06 \\
ACS5-034 & 46 & TX & HS diploma & Carpenters & 1.8570 & T06 \\
ACS5-035 & 49 & WA & GED & Childcare workers & 3.3229 & T11 \\
ACS5-036 & 54 & CA & Master & Food service managers & 4.3111 & T04 \\
ACS5-037 & 34 & FL & College 1y+ & Electricians & 1.0732 & T06 \\
ACS5-038 & 38 & NY & Associate & Miscellaneous agricultural workers & 3.9452 & T09 \\
ACS5-039 & 41 & TX & HS diploma & Licensed practical and licensed vocational nurses & 1.8833 & T01 \\
ACS5-040 & 34 & WA & Grade 11 & Construction laborers & 1.1609 & T06 \\
ACS5-041 & 23 & CA & College 1y+ & Fast food and counter workers & 3.7252 & T04 \\
ACS5-042 & 49 & FL & Grade 11 & Shipping, receiving, and inventory clerks & 2.5906 & T03 \\
ACS5-043 & 61 & NY & HS diploma & Carpenters & 1.6191 & T06 \\
ACS5-044 & 30 & TX & GED & Security guards and gambling surveillance officers & 3.1478 & T07 \\
ACS5-045 & 27 & WA & HS diploma & Chefs and head cooks & 4.3293 & T04 \\
ACS5-046 & 18 & CA & HS diploma & Cooks & 3.9253 & T04 \\
ACS5-047 & 35 & FL & GED & Police officers & 2.0385 & T07 \\
ACS5-048 & 24 & NY & Associate & Barbers & 4.9540 & T12 \\
ACS5-049 & 33 & TX & College 1y+ & Heating, air conditioning, and refrigeration mechanics and installers & 2.2024 & T06 \\
ACS5-050 & 40 & WA & Grade 11 & Home health aides & 3.5859 & T01 \\
\end{longtable}
\endgroup
\begin{figure}[!tp]
\centering
\includegraphics[width=\linewidth]{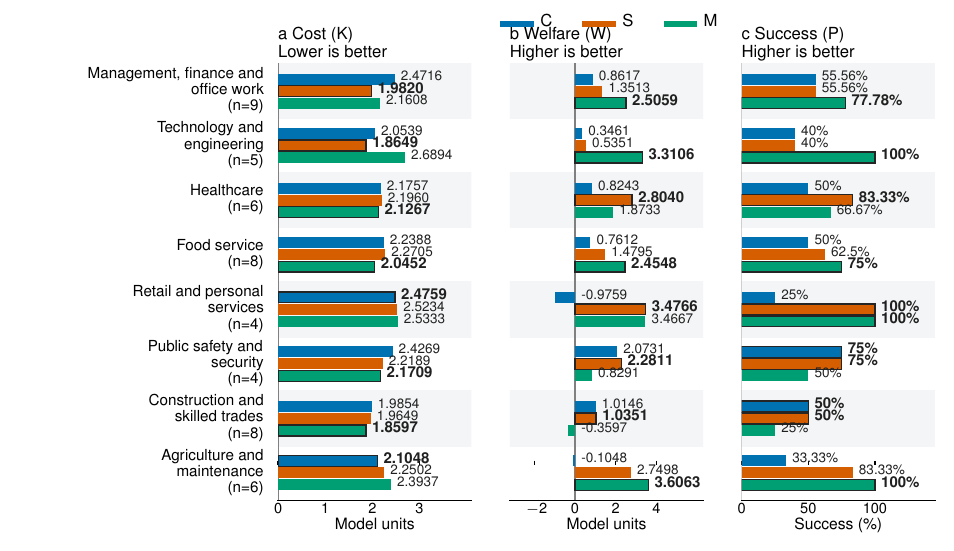}
\caption{Comparison of cost, welfare and provision across occupational groups (Table~\ref{vten:tab:occupation}). Within each row, bars are ordered C, S, M from top to bottom and start at zero. Labels give all 72 values: mean $K$ and $W$ to four decimals and success $P$ in percent. Dark outlines and bold labels mark the lowest $K$ or highest $W,P$ within each group; all ties are highlighted. Highlighting denotes a numerical ranking, not statistical significance. All 50 profiles and all valid failures are retained. Welfare follows $W=6I-K$; comparisons are descriptive.}\label{vten:fig:nature-table08_occupation}
\Description{Three aligned panels compare cost, welfare and provision for eight occupational groups. Each group includes C, S and M; numerical labels reproduce the table values.}

\end{figure}

\begin{table}[tp]
\caption{Final outcomes by educational attainment. $n$ is the number of profiles per mechanism.}\label{vten:tab:education}
\centering
\small
\setlength{\tabcolsep}{3pt}
\begin{tabular}{@{}>{\raggedright\arraybackslash}p{33mm}r*{9}{r}@{}}
\toprule
 & & \multicolumn{3}{c}{C} & \multicolumn{3}{c}{S} & \multicolumn{3}{c}{M} \\
\cmidrule(lr){3-5}\cmidrule(lr){6-8}\cmidrule(l){9-11}
Group & $n$ & $K\downarrow$ & $W\uparrow$ & $P\uparrow$ & $K\downarrow$ & $W\uparrow$ & $P\uparrow$ & $K\downarrow$ & $W\uparrow$ & $P\uparrow$ \\
\midrule
No schooling completed & 1 & 1.9999 & -1.9999 & 0\% & 2.2669 & 3.7331 & 100\% & 2.2827 & 3.7173 & 100\% \\
Grade 9 & 1 & 1.9919 & -1.9919 & 0\% & 1.9551 & -1.9551 & 0\% & 1.7549 & -1.7549 & 0\% \\
Grade 11 & 3 & 2.2483 & -0.2483 & 33.33\% & 2.0122 & 1.9878 & 66.67\% & 2.1277 & 1.8723 & 66.67\% \\
High school diploma & 11 & 2.0008 & 0.1811 & 36.36\% & 2.2036 & 2.7055 & 81.82\% & 2.1130 & 2.2506 & 72.73\% \\
GED or alternative credential & 3 & 2.3947 & 1.6053 & 66.67\% & 2.2708 & 3.7292 & 100\% & 2.1281 & 1.8719 & 66.67\% \\
Some college, less than one year & 3 & 2.7702 & 1.2298 & 66.67\% & 1.6244 & 0.3756 & 33.33\% & 1.9054 & 2.0946 & 66.67\% \\
Some college, at least one year; no degree & 6 & 2.2627 & 2.7373 & 83.33\% & 2.1495 & 1.8505 & 66.67\% & 2.2201 & 0.7799 & 50\% \\
Associate's degree & 8 & 2.1661 & -0.6661 & 25\% & 2.3669 & 1.3831 & 62.5\% & 2.1571 & 3.0929 & 87.5\% \\
Bachelor's degree & 10 & 2.2777 & 0.7223 & 50\% & 1.9485 & 1.6515 & 60\% & 2.3339 & 1.8661 & 70\% \\
Master's degree & 3 & 2.4171 & 1.5829 & 66.67\% & 2.4017 & 1.5983 & 66.67\% & 2.6710 & 3.3290 & 100\% \\
Professional degree beyond bachelor & 1 & 2.4279 & 3.5721 & 100\% & 1.9844 & -1.9844 & 0\% & 2.3787 & 3.6213 & 100\% \\
\bottomrule
\end{tabular}
\end{table}
\begin{table}[tp]
\caption{Final outcomes by source-record state.}\label{vten:tab:state}
\centering
\small
\setlength{\tabcolsep}{3pt}
\begin{tabular}{@{}>{\raggedright\arraybackslash}p{33mm}r*{9}{r}@{}}
\toprule
 & & \multicolumn{3}{c}{C} & \multicolumn{3}{c}{S} & \multicolumn{3}{c}{M} \\
\cmidrule(lr){3-5}\cmidrule(lr){6-8}\cmidrule(l){9-11}
Group & $n$ & $K\downarrow$ & $W\uparrow$ & $P\uparrow$ & $K\downarrow$ & $W\uparrow$ & $P\uparrow$ & $K\downarrow$ & $W\uparrow$ & $P\uparrow$ \\
\midrule
California & 10 & 2.3296 & 0.6704 & 50\% & 2.2813 & 0.7187 & 50\% & 2.3912 & 2.4088 & 80\% \\
Florida & 10 & 1.9128 & -0.1128 & 30\% & 2.0603 & 0.3397 & 40\% & 2.1289 & 2.0711 & 70\% \\
New York & 10 & 2.2539 & 0.1461 & 40\% & 2.3028 & 3.0972 & 90\% & 2.4144 & 2.9856 & 90\% \\
Texas & 10 & 2.2651 & 1.3349 & 60\% & 1.9686 & 2.2314 & 70\% & 1.8992 & 1.1008 & 50\% \\
Washington & 10 & 2.3987 & 1.2013 & 60\% & 2.0561 & 2.7439 & 80\% & 2.1735 & 2.0265 & 70\% \\
\bottomrule
\end{tabular}
\end{table}
\begin{table}[tp]
\caption{Final outcomes by broad occupational group. $n$ is the number of matched profiles per mechanism.}\label{vten:tab:occupation}
\centering
\small
\setlength{\tabcolsep}{3pt}
\begin{tabular}{@{}>{\raggedright\arraybackslash}p{33mm}r*{9}{r}@{}}
\toprule
 & & \multicolumn{3}{c}{C} & \multicolumn{3}{c}{S} & \multicolumn{3}{c}{M} \\
\cmidrule(lr){3-5}\cmidrule(lr){6-8}\cmidrule(l){9-11}
Group & $n$ & $K\downarrow$ & $W\uparrow$ & $P\uparrow$ & $K\downarrow$ & $W\uparrow$ & $P\uparrow$ & $K\downarrow$ & $W\uparrow$ & $P\uparrow$ \\
\midrule
Management, finance and office work & 9 & 2.4716 & 0.8617 & 55.56\% & 1.9820 & 1.3513 & 55.56\% & 2.1608 & 2.5059 & 77.78\% \\
Technology and engineering & 5 & 2.0539 & 0.3461 & 40\% & 1.8649 & 0.5351 & 40\% & 2.6894 & 3.3106 & 100\% \\
Healthcare & 6 & 2.1757 & 0.8243 & 50\% & 2.1960 & 2.8040 & 83.33\% & 2.1267 & 1.8733 & 66.67\% \\
Food service & 8 & 2.2388 & 0.7612 & 50\% & 2.2705 & 1.4795 & 62.5\% & 2.0452 & 2.4548 & 75\% \\
Retail and personal services & 4 & 2.4759 & -0.9759 & 25\% & 2.5234 & 3.4766 & 100\% & 2.5333 & 3.4667 & 100\% \\
Public safety and security & 4 & 2.4269 & 2.0731 & 75\% & 2.2189 & 2.2811 & 75\% & 2.1709 & 0.8291 & 50\% \\
Construction and skilled trades & 8 & 1.9854 & 1.0146 & 50\% & 1.9649 & 1.0351 & 50\% & 1.8597 & -0.3597 & 25\% \\
Agriculture and maintenance & 6 & 2.1048 & -0.1048 & 33.33\% & 2.2502 & 2.7498 & 83.33\% & 2.3937 & 3.6063 & 100\% \\
\bottomrule
\end{tabular}
\end{table}
\begin{figure}[tp]
\centering
\includegraphics[width=\linewidth]{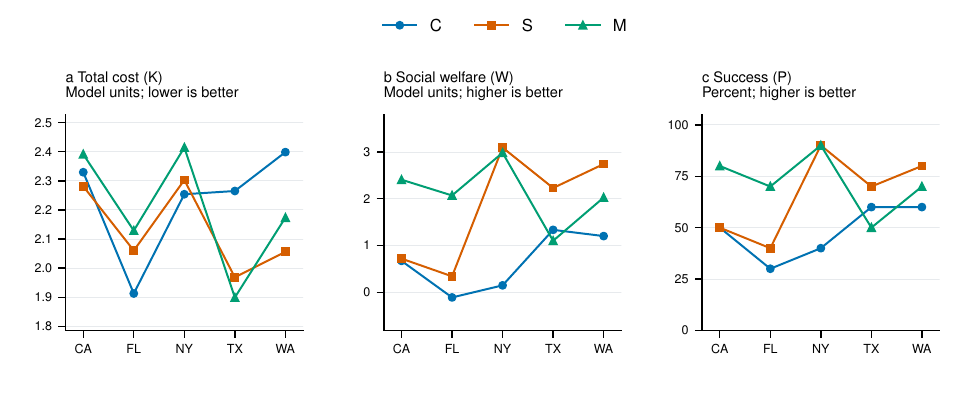}
\caption{State-level mean cost $K$, welfare $W$ and provision success $P$ (Table~\ref{vten:tab:state}); ten paired profiles per state. Circle, square and triangle markers denote C, S and M. Lines are categorical visual guides, not continuous trends. All valid failures are included; $W=6I-K$. Comparisons are descriptive, not causal.}\label{vten:fig:nature-table06_state}
\Description{Three line panels compare total cost, social welfare and provision across California, Florida, New York, Texas and Washington for C, S and M.}

\end{figure}

\clearpage
\section{Proofs Omitted from Section~\ref{sec:model}}
\subsection{Strict BNE Proof for Proposition~\ref{prop:floor} (Floor Equilibrium)}
\label{app:floor-proof}

\begin{proof}
Fix an arbitrary user \(i\) with cost type \(c_i\in[\underline c,\bar c]\) and suppose all other users \(j\neq i\) follow the strategy
\[
e_j = e^{0}(c_j) = \min\!\Big\{\frac{p}{c_j},\,1\Big\}.
\]
User \(i\) chooses an action \(e_i\in[0,1]\) without knowing the realizations \(\{c_j\}_{j\neq i}\). His ex‑post payoff, given the full cost vector, is
\[
u_i(e_i,\mathbf c_{-i}) = V\cdot\mathbf 1_{\{e_i+\sum_{j\neq i}e^{0}(c_j)\ge X\}} + p e_i - \frac{c_i}{2}e_i^{2}.
\]

For all \(c_j\in[\underline c,\bar c]\) we have \(e^{0}(c_j)\le \min\{p/\underline c,\,1\}\) because \(p/c_j\le p/\underline c\) and the minimum with \(1\) is non‑decreasing in the argument. Hence, for every realization of \(\mathbf c_{-i}\),
\[
\sum_{j\neq i}e^{0}(c_j) \le (n-1)\min\!\Big\{\frac{p}{\underline c},\,1\Big\}.
\]
If the condition
\[
1+(n-1)\min\!\Big\{\frac{p}{\underline c},\,1\Big\} < X
\]
holds, then even the maximal possible deviation \(e_i=1\) together with the largest possible aggregate contribution from the other users yields
\[
e_i+\sum_{j\neq i}e^{0}(c_j) \le 1 + (n-1)\min\!\Big\{\frac{p}{\underline c},\,1\Big\} < X.
\]
Therefore, for every \(e_i\in[0,1]\) and every \(\mathbf c_{-i}\) the indicator \(\mathbf 1_{\{e_i+\sum_{j\neq i}e^{0}(c_j)\ge X\}}\) is identically zero. Consequently, user \(i\)’s payoff simplifies to the deterministic function
\[
u_i(e_i) = p e_i - \frac{c_i}{2}e_i^{2},
\]
independent of others’ types and of user \(i\)’s beliefs about them. The function \(e_i\mapsto p e_i-\frac{c_i}{2}e_i^{2}\) is strictly concave (its second derivative is \(-c_i<0\)), so the unique maximizer on \([0,1]\) is found by setting the first derivative to zero:
\[
\frac{\partial}{\partial e_i}\Big(p e_i - \frac{c_i}{2}e_i^{2}\Big) = p - c_i e_i = 0
\quad\Longrightarrow\quad e_i^{*} = \frac{p}{c_i}.
\]
When \(p/c_i\le 1\) this interior solution is feasible and indeed the optimal choice. If \(p/c_i>1\) the derivative is positive on the whole interval \([0,1]\), so the maximum is attained at the corner \(e_i=1\). Hence the best response of user \(i\) for any \(c_i\) is exactly
\[
e_i^{*} = \min\!\Big\{\frac{p}{c_i},\,1\Big\} = e^{0}(c_i).
\]

Thus, playing \(e_i = e^{0}(c_i)\) is a best response for every type \(c_i\) and every user \(i\). The strategy profile \(\{e^{0}(c_i)\}_{i=1}^{n}\) is therefore a Bayesian Nash equilibrium.
\end{proof}

\subsection{Proof of Proposition~\ref{prop:Cutoff} (Cutoff Equilibrium)}
\label{app:cutoff-proof}

\noindent\textbf{Calibration of $\tilde{g}(a)$.}
Fix a candidate cutoff $a$. Each of the other $n-1$ users independently participates (contributes $\tilde{g}(a)$) with probability $q = F(a)$, or contributes floor $p/c$ with probability $1-q$. Setting the expected aggregate equal to $X$ when exactly $m+1$ users (including $i$) participate:
\[
\begin{aligned}
&(m+1)\tilde{g}(a)\\
&\quad +(n-1-m)\mu_0(a)=X\\
&\Longrightarrow\quad
\tilde{g}(a)
=\frac{X-(n-1-m)\mu_0(a)}{m+1}.
\end{aligned}
\]

\medskip
\noindent\textbf{Pivotal reduction under binomial concentration.}
Replace the random floor aggregate $A_k = \sum_{j=1}^k p/c_j'$ by its expectation $k\mu_0(a)$. This is asymptotically exact by the law of large numbers and standard in the threshold public-goods literature. Conditional on $h$ other participants, provision then requires $e_i \geq t_h(a)$ where $t_h(a) = (m+1-h)\tilde{g}(a) + (h-m)\mu_0(a)$. In particular, $t_m = \tilde{g}$: provision requires user $i$ to contribute the full participation share when exactly $m$ others participate. The pivotal probability is
\[
  \Delta B^{\mathrm{bin}}(a) = \tbinom{n-1}{m}\, q^m(1-q)^{n-1-m}.
\]

Under this approximation, the provision probability $B(e_i, a_p)$ is a step function in $e_i$ with a jump at $e_i = \tilde{g}(a_p)$. Below $\tilde{g}$, the provision probability is $\Pr(h \geq m+1)$; at or above $\tilde{g}$, it is $\Pr(h \geq m) = \Pr(h \geq m+1) + \Delta B^{\mathrm{bin}}(a_p)$. On each constant region, the payoff $pe_i - c_i e_i^2/2$ is maximized at the floor $p/c_i$. Hence the only candidates for a best response are the floor $p/c_i$ and $\tilde{g}(a_p)$.

\medskip
\noindent\textbf{Indifference condition and cutoff structure.}
The ex-ante strategy prescribes: contribute $\tilde{g}(a_p)$ if $c \leq a_p$, floor $p/c$ if $c > a_p$. For this to be an ex-ante BNE, each realized type must prefer its prescribed action. The net benefit of participating at $\tilde{g}(a_p)$ rather than flooring, for a user who realizes type $c$, is
\[
  varphi(c) = V \cdot \Delta B^{\mathrm{bin}}(a_p) - \Gamma^*(c, p, \tilde{g}(a_p)).
\]
The first term is the expected value of being pivotal. It depends on the equilibrium cutoff $a_p$ (through others' strategies) but \emph{not} on the individual type $c$. The second term is the net participation cost $\Gamma^*(c, p, \tilde{g}) = (c\tilde{g} - p)^2/(2c)$, which is strictly increasing in $c$:
\[
  \frac{\partial \Gamma^*}{\partial c} = \frac{c^2 \tilde{g}^2 - p^2}{2c^2} > 0 \quad \text{whenever } c\tilde{g} > p.
\]
Since $c \geq \underline{c}$ and $\tilde{g}(a_p) > p/\underline{c} \geq p/c$, this condition holds throughout the support. Therefore $\varphi(c)$ is strictly decreasing: a constant minus a strictly increasing function. At $c = a_p$: $\varphi(a_p) = \Phi(a_p) = 0$ by construction. Hence:
\begin{itemize}
  \item $c < a_p$: $\varphi(c) > 0$, participating is strictly optimal.
  \item $c > a_p$: $\varphi(c) < 0$, flooring is strictly optimal.
  \item $c = a_p$: indifferent; convention assigns $\tilde{g}(a_p)$.
\end{itemize}
Every realized type plays a best response to the common cutoff strategy, confirming that proposition\ref{prop:Cutoff} is an ex-ante symmetric BNE.

\medskip
\noindent\textbf{Existence of $a_p$.}
Define 
\[
\Phi(a) = V \cdot \Delta B^{\mathrm{bin}}(a) - \Gamma^*(a, p, \tilde{g}(a))
\]
The binomial term is maximized at
\[
a^* = F^{-1}(m/(n-1)),
\]
where $\Phi(a^*) > 0$ for $V > \underline{V}$. At $a = \bar{c}$: $q = 1$ gives $\Delta B^{\mathrm{bin}} = 0$, so $\Phi(\bar{c}) = -\Gamma^* < 0$. By continuity (from $f \in C^1$, $f > 0$), the intermediate value theorem gives a root $a_p \in (a^*, \bar{c})$.

\medskip
\noindent\textbf{Uniqueness under \(\Delta\)-quasiconcavity.}
Rewrite the cutoff condition \(\Phi(a)=0\) as
\[
V=R(a),
\qquad
R(a)=
\frac{\Gamma^*(a,p,\tilde g(a))}
{\Delta B^{\mathrm{bin}}(a)}.
\]
The numerator \(\Gamma^*(a,p,\tilde g(a))\) is strictly increasing in \(a\), as
shown above. Under the \(\Delta\)-quasiconcavity condition,
\(\Delta B^{\mathrm{bin}}(a)\) is single-peaked, with peak
\[
a^*=F^{-1}\!\left(\frac{m}{n-1}\right).
\]
For \(a>a^*\), the numerator increases while the denominator decreases. Hence
\(R(a)\) is strictly increasing on this provision-relevant branch. For
\(a<a^*\), \(R(a)\) may decrease, so \(R(a)\) is U-shaped. The existence
argument selects a cutoff \(a_p>a^*\). Therefore, on the selected branch, the
equation
\[
V=R(a)
\]
has exactly one solution.

At the root, \(R'(a_p)>0\). Since
\[
\Phi(a)
=
\Delta B^{\mathrm{bin}}(a)\bigl[V-R(a)\bigr],
\]
we obtain
\[
\Phi'(a_p)
=
-R'(a_p)\Delta B^{\mathrm{bin}}(a_p)
<0.
\]
Thus the cutoff is locally stable. Since \(R'(a_p)\neq 0\), the implicit function
theorem implies that \(a_p\) is \(C^1\)-smooth in \((p,V)\).

\section{Derivation of the Optimal Assignment Structure}
\label{app:backstop-pool}

This appendix provides the formal foundation for the floor-plus-backstop assignment described in Section~\ref{sec:withdrawal}.  We take the full-revelation equilibrium and the provider's rationality (Propositions~\ref{prop:revelation} and~\ref{prop:Assignment}) as given.

Fix a realized cost vector, ordered as \(c_{(1)}\le c_{(2)}\le\cdots\le c_{(n)}\).  
Every user \(j\) has a privately optimal floor contribution \(e^0(c_j)=\min\{p/c_j,1\}\), which he would retain regardless of whether the model improves.  The provider's problem is to choose an assignment \(\{g_j\}_{j=1}^{n}\) with \(g_j\in[0,1]\) such that in the subsequent withdrawal game the aggregate effective contribution reaches the threshold \(X\) whenever a positive assignment is made. Conditional on positive provision, the provider does not assign excess contribution beyond the threshold, since additional assigned contribution only increases privacy burden and subsidy expenditure without increasing the provision benefit.

As argued in the main text, users who are not pivotal for reaching the threshold can be assigned exactly their floor amounts, leaving a residual demand 
\[
D_K = X - \sum_{j\notin K} e^0(c_j)
\]
to be covered by a set \(K\) of backstoppers, where \(K\subseteq\{1,\dots,n\}\) and \(|K|=k\).  
The feasibility condition \(D_K\le k\) follows because each backstopper can retain at most one unit of data.

We now show that, for a given residual demand \(D_K\), the provider minimizes the tightest participation constraint among backstoppers by selecting the \(k\) users with the smallest costs \(c_j\).

\begin{lemma}[Swap argument]
\label{lem:swap}
Let \(i,j\) be two users with \(c_i < c_j\). Suppose a candidate assignment requires \(i\) to contribute at or below his floor \((g_i\le e^0(c_i))\) while requiring \(j\) to contribute strictly above his floor \((g_j > e^0(c_j))\).  Then transferring a sufficiently small amount of the above-floor burden from \(j\) to \(i\) does not tighten any backstopper's participation constraint, and strictly relaxes it if \(i\) was initially at the floor.
\end{lemma}

\begin{proof}
For a user with cost \(c\) and required effective contribution \(d>e^0(c)\), the incremental participation cost is 
\[
\Gamma^*(c,p,d)=\frac{(cd-p)^2}{2c}.
\]
For fixed \(d>p/c\) we have 
\[
\frac{\partial\Gamma^*(c,p,d)}{\partial c}= \frac{d^2}{2}-\frac{p^2}{2c^2}>0,
\]
so \(\Gamma^*\) is strictly increasing in \(c\).  Now reduce \(g_j\) by a small \(\varepsilon>0\) and increase \(g_i\) by the same \(\varepsilon\), keeping the aggregate contribution unchanged.  The transfer replaces part of \(j\)'s above-floor burden with an equal amount of \(i\)'s above-floor burden.  Because \(\Gamma^*\) is strictly increasing in cost, the maximum participation cost in the pool does not increase; it strictly decreases if \(i\) was previously at the floor.  Iterating such transfers yields the conclusion.
\end{proof}

Lemma~\ref{lem:swap} implies that, for any fixed residual demand \(D_K\), the provider weakly relaxes the backstoppers' participation constraints by placing above-floor burdens on lower-cost users before higher-cost users. Hence, if a feasible positive assignment exists with a backstop pool of size \(k\), there exists a feasible assignment using the \(k\) lowest-cost users as backstoppers.

Importantly, the assignment \(g_j\) is a target contribution rather than a maximal authorization. Therefore, the provider does not assign \(g_j=1\) to every backstopper. Instead, for \(j\in K\), the provider assigns protocol-specific targets \(g_j=e_j^*\in[e^0(c_j),1]\) satisfying
\[
\sum_{j\in K} e_j^* = D_K.
\]
The exact split of \(D_K\) across backstoppers depends on the withdrawal protocol: Protocol \(\mathcal{S}\) uses the equal-marginal-cost target allocation, while Protocol \(\mathcal{M}\) uses the small-first sequential construction described in Appendix~\ref{app:containment}.

This floor-plus-backstop structure is valid for both withdrawal protocols \(\mathcal{S}\) and \(\mathcal{M}\). The provider designs the assignment at Stage~1, before the withdrawal stage begins. The optimality argument above relies only on the fact that every non-backstopper has a dominant strategy to retain exactly \(e^0(c_j)\) regardless of the protocol, a property that holds under both simultaneous and small-first withdrawal. Hence the floor-plus-backstop form is protocol-independent, while the allocation of the residual demand among backstoppers is protocol-specific.

\section{Protocol Comparison and Containment}
\label{app:containment}

This appendix provides the formal derivations omitted from Section~\ref{sec:comparison}. After justifying the small‑first ordering of Protocol~$\mathcal M$ against the alternative large‑first protocol, we establish the optimality of the equal‑marginal‑cost allocation underlying Protocol~$\mathcal S$ and give the complete proof of Proposition~\ref{prop:Containment} (containment of provision regions).  We then present a fully detailed illustrative example.  The final part compares the total privacy costs of the two protocols when both succeed.

\subsection{Small-First vs.\ Large-First Ordering}
\label{app:small-large}

The main text defines Protocol~\(\mathcal M\) as the small-first withdrawal protocol, among equally assigned backstoppers, higher-cost users move earlier and the lowest-cost user moves last. We compare it with the reverse ordering.

For a user with cost \(c\), define the largest effective contribution he is willing to
retain on a provision path:
\[
\bar d(c;V,p)
=
\min\left\{
1,\frac{p+\sqrt{2cV}}{c}
\right\}.
\]
This follows from the participation constraint
\[
V\ge \Gamma^*(c,p,d)
=
\frac{(cd-p)^2}{2c}.
\]
Thus \(\bar d(c;V,p)\) is the largest residual gap that type \(c\) can fill while still
preferring provision to withdrawal to the floor. Since
\[
\frac{p+\sqrt{2cV}}{c}
=
\frac{p}{c}+\sqrt{\frac{2V}{c}},
\]
\(\bar d(c;V,p)\) is weakly decreasing in \(c\). Lower-cost backstoppers can therefore
fill weakly larger residual gaps.

\paragraph{Large-first protocol.}
Define the alternative large-first protocol \(\mathcal L\) as the reverse of \(\mathcal M\). Backstoppers move in decreasing order of assigned amounts, with ties among equally assigned backstoppers broken by increasing cost. Since all backstoppers receive the same notional assignment \(g_{(j)}=1\), the order under \(\mathcal L\) is $(1)\rightarrow(2)\rightarrow\cdots\rightarrow(k)$, so the lowest-cost backstopper moves first and the highest-cost backstopper moves last. Let
\[
E_F=\sum_{j\notin K}e^0(c_j)
\]
be the aggregate floor contribution from non-backstoppers. For a running aggregate \(E\), define the large-first continuation-success indicator by
\[
S_{k+1}^{\mathcal L}(E)=\mathbf 1\{E\ge X\},
\]
and, for \(j=1,\ldots,k\),
\[
\begin{aligned}
S_j^{\mathcal L}(E)
&=
\mathbf 1
\left\{
\exists d_j\in
\left[e^0(c_{(j)}),
\bar d(c_{(j)};V,p)\right]\right.\\
&\hspace{4em}\left.
\ \text{s.t.}\
S_{j+1}^{\mathcal L}(E+d_j)=1
\right\}.
\end{aligned}
\]
Therefore \(S_j^{\mathcal L}(E)=1\) means that, starting from current aggregate \(E\)
and from backstopper \((j)\), the remaining large-first chain can still reach the
threshold while satisfying all participation constraints. The large-first provision region
is
\[
\Omega_P^{\mathcal L}(p)
=
\left\{
\mathbf c:
S_1^{\mathcal L}(E_F)=1
\right\}.
\]

\begin{prop}[Small-first weakly dominates large-first]
\label{prop:small-vs-large}
For any subsidy \(p\),
\[
\Omega_P^{\mathcal L}(p)\subseteq \Omega_P^{\mathcal M}(p).
\]
The inclusion is strict for an open set of heterogeneous multi-backstopper cost
realizations.
\end{prop}

\begin{proof}
Fix a cost realization \(\mathbf c\in\Omega_P^{\mathcal L}(p)\). Then there exists a
large-first feasible retention path
\[
(d_1^{\mathcal L},\ldots,d_k^{\mathcal L})
\]
such that
\[
E_F+\sum_{j=1}^k d_j^{\mathcal L}\ge X,
\qquad
d_j^{\mathcal L}\le \bar d(c_{(j)};V,p)
\quad \text{for all }j.
\]
Because \(\bar d(c;V,p)\) is weakly decreasing in \(c\), any effective contribution that
is feasible for a higher-cost backstopper is also feasible for a lower-cost backstopper.
Thus the feasible large-first retention profile can be rearranged so that larger residual
burdens are assigned to weakly lower-cost users.

This rearrangement is exactly what the small-first order permits. Under \(\mathcal M\),
higher-cost backstoppers move earlier, while lower-cost backstoppers move later and can
absorb the residual gap left by earlier movers. In particular, the final mover under
large-first is \(c_{(k)}\), the highest-cost backstopper with the smallest residual-filling
capacity, whereas the final mover under small-first is \(c_{(1)}\), the lowest-cost
backstopper with the largest residual-filling capacity. Replacing the weakest final safety
net with the strongest one cannot destroy a feasible continuation chain. Therefore
\(\mathbf c\in\Omega_P^{\mathcal M}(p)\).

To see that the inclusion is generically strict, consider \(k=2\) with
\(c_{(1)}<c_{(2)}\). Since \(\bar d(c;V,p)\) is decreasing in \(c\), there exist parameter
values and a residual gap \(R\) such that
\[
\bar d(c_{(2)};V,p)<R\le \bar d(c_{(1)};V,p).
\]
Choose the residual demand so that the last mover must fill \(R\). Under small-first, the
last mover is the low-cost backstopper \((1)\), who can fill \(R\). Under large-first, the
last mover is the high-cost backstopper \((2)\), who cannot fill \(R\). Hence the
realization succeeds under \(\mathcal M\) but fails under \(\mathcal L\). Such inequalities
hold on an open set whenever costs are heterogeneous, so the containment is generically
strict.
\end{proof}

\subsection{Derivation of the Simultaneous-Withdrawal Target}
\label{app:equal-exposure}

This subsection derives the target allocation used in Protocol~\(\mathcal S\). We do not
repeat the definition of the simultaneous provision region, which is given in
\eqref{eq:S-best-alloc}--\eqref{eq:S-region}. Instead, we explain why the target has
the inverse-cost form in \eqref{eq:S-best-alloc} and why the participation threshold in
\eqref{eq:SV} is the relevant equilibrium condition.

Under Protocol~\(\mathcal S\), all backstoppers choose their retentions simultaneously.
Thus a backstopper cannot condition his action on the realized withdrawals of other backstoppers. Moreover, although the provider observes the full cost vector, an individual
backstopper observes only his own assignment and does not know the other backstoppers' costs. Hence the provider must communicate a personalized target retention to each backstopper before the withdrawal stage.

Let \(D_K\) be the residual demand that must be covered by the backstop pool \(K\).
A simultaneous target vector must satisfy
\[
\sum_{j\in K} e_j = D_K.
\]
The target in \eqref{eq:S-best-alloc} is obtained by equalizing the exposure term
\[
T_j \equiv c_j e_j
\]
across backstoppers. If \(T_j=T\) for all \(j\in K\), then \(e_j=T/c_j\). Substituting
this into the residual-demand constraint gives
\[
D_K
=
\sum_{j\in K}\frac{T}{c_j}
=
T\sum_{j\in K}\frac{1}{c_j}
=
T\bar C_K.
\]
Therefore
\[
T = \frac{D_K}{\bar C_K}, \qquad 
e_j^* = \frac{T}{c_j} = \frac{D_K}{\bar C_K c_j},
\]
which is exactly the target in \eqref{eq:S-best-alloc}. Thus the simultaneous target assigns larger effective retentions to lower-cost users and smaller retentions to higher-cost users.

This target also gives a simple way to compute the binding participation constraint. If all backstoppers retain their targets, then the aggregate contribution reaches the
threshold. If backstopper \(j\) deviates to his floor contribution, aggregate contribution falls below the threshold, so he is pivotal. Therefore user \(j\)'s comparison is between
retaining \(e_j^*\) and receiving the model-improvement value, or withdrawing to his floor and losing provision.

The payoff from retaining \(e_j^*\) is $V+p e_j^*-\frac{c_j}{2}(e_j^*)^2,$
whereas the floor payoff is ${p^2}/2c_j$. Thus retaining the target is optimal if and only if
\[
V+p e_j^*-\frac{c_j}{2}(e_j^*)^2
\ge
\frac{p^2}{2c_j}.
\]
Rearranging gives
\[
V
\ge
\frac{c_j}{2}(e_j^*)^2-p e_j^*+\frac{p^2}{2c_j}
=
\frac{(c_j e_j^*-p)^2}{2c_j}
=
\Gamma^*(c_j,p,e_j^*).
\]
Hence the simultaneous provision equilibrium exists precisely when every backstopper's participation constraint is satisfied:
\[
V\ge \max_{j\in K}\Gamma^*(c_j,p,e_j^*).
\]

Using \(c_j e_j^*=D_K/\bar C_K\), the binding threshold can be written as
\[
\max_{j\in K}
\frac{(D_K/\bar C_K-p)^2}{2c_j}.
\]
Since the users in \(K\) are ordered by cost and \(c_{(1)}\) is the smallest cost in the
backstop pool, the maximum is attained at \(c_{(1)}\). Therefore,
\[
\bar V^{\mathcal S}(\mathbf c,p)
=
\frac{(D_K/\bar C_K-p)^2}{2c_{(1)}}.
\]
This yields the provision condition stated in \eqref{eq:SV}.

If the interior target in \eqref{eq:S-best-alloc} violates the feasibility bounds \(e_j\in[e^0(c_j),1]\), the same derivation applies after fixing the constrained users at their bounds. Users with targets below their floors are set to \(e^0(c_j)\), users with targets above one are set to \(1\), and the remaining residual demand is reallocated among the unconstrained backstoppers by equalizing \(c_j e_j\) over that remaining set.

\subsection{Proof of Proposition~\ref{prop:Containment}}
\label{app:proof-containment}

We first recall the provision regions characterized in the main text. For the
single-backstopper case, the two protocols coincide, as stated in
\eqref{eq:regin1}. We therefore focus on the multi-backstopper case \(D_K>1\).

For the small-first protocol \(\mathcal M\), the backward-induction construction starts from the terminal success condition in \eqref{eq:S1}. The resulting provision region is
given by
\[
\Omega_P^{\mathcal M}(p)
=
\bigl\{
(c_1,\dots,c_n):
S_k(c_{(k)},\dots,c_{(1)},E_F)=1
\bigr\},
\]
as in \eqref{eq:regin2}, where
\[
E_F=\sum_{j\notin K}e^0(c_j)
\]
is the floor aggregate from non-backstoppers.

For the simultaneous protocol \(\mathcal S\), the provider uses the target allocation in \eqref{eq:S-best-alloc}. The simultaneous provision equilibrium exists exactly when the
participation threshold in \eqref{eq:SV} is satisfied. Equivalently, the provision region is
\[
\Omega_P^{\mathcal S}(p)
=
\bigl\{
(c_1,\dots,c_n):
V\ge \bar V^{\mathcal S}(\mathbf c,p)
\bigr\},
\]
as in \eqref{eq:S-region}. With these definitions, we prove the containment result. The key point is that any target-retention vector that sustains simultaneous provision also defines a feasible continuation path under the small-first sequential protocol.

\begin{proof}
Fix a cost realization
\[
\mathbf c\in \Omega_P^{\mathcal S}(p).
\]
By the definition of the simultaneous provision region, there exists a target-retention vector \(\{e_j^*\}_{j\in K}\) such that
\[
\sum_{j\in K} e_j^* = D_K,
\]
and every backstopper's participation constraint is satisfied:
\[
V \ge \Gamma^*(c_j,p,e_j^*), \quad \forall j \in K.
\]
Together with the non-backstoppers' floor contributions, this vector exactly reaches the threshold:
\[
E_F+\sum_{j\in K}e_j^*
=
X.
\]

We now construct a feasible small-first continuation path. Under \(\mathcal M\), the backstoppers move in the order $(k),(k-1),\ldots,(1)$, so higher-cost backstoppers move earlier and lower-cost backstoppers move later. Consider the candidate path in which each backstopper \((j)\) retains \(e_{(j)}^*\). Verify this path by backward induction. Start with the final mover \((1)\). If all earlier movers \((k),\ldots,(2)\) have retained their candidate amounts, then the running aggregate before \((1)\)'s move is
\[
E_{<1}
=
E_F+\sum_{\ell=2}^{k}e_{(\ell)}^*.
\]
Therefore the residual gap faced by the last mover is
\[
X-E_{<1}
=
X-E_F-\sum_{\ell=2}^{k}e_{(\ell)}^*
=
e_{(1)}^*.
\]
Since \(e_{(1)}^*\le 1\) and $V\ge\Gamma^*(c_{(1)},p,e_{(1)}^*)$, the last mover is willing to fill the residual gap. Hence the terminal success condition in the small-first chain is satisfied.

Now suppose that, for some \(j\ge 2\), the lower-cost continuation $(j-1),(j-2),\ldots,(1)$ can successfully cover the residual demand assigned to it by the candidate vector. We show that backstopper \((j)\) can choose \(e_{(j)}^*\) and leave a feasible continuation. Before \((j)\)'s move, the running aggregate is
\[
E_{<j}
=
E_F+\sum_{\ell=j+1}^{k} e_{(\ell)}^*,
\]
where the sum is over higher-cost backstoppers who have already moved. If \((j)\) retains
\(e_{(j)}^*\), the remaining residual gap for later movers is
\[
X-E_{<j}-e_{(j)}^*
=
X-E_F-\sum_{\ell=j}^{k}e_{(\ell)}^*
=
\sum_{\ell=1}^{j-1}e_{(\ell)}^*.
\]
This is exactly the total amount assigned by the candidate vector to the later, lower-cost backstoppers. By the induction hypothesis, these later movers can cover this residual gap while satisfying their participation constraints. Moreover, backstopper \((j)\)'s own participation constraint is satisfied because
\[
V\ge \Gamma^*(c_{(j)},p,e_{(j)}^*).
\]
Thus \(e_{(j)}^*\) is a feasible action that keeps the continuation chain successful.

By induction, the small-first chain succeeds from the initial state \(E_F\). Therefore the
success indicator in the main text satisfies
\[
S_k(c_{(k)},\ldots,c_{(1)},E_F)=1,
\]
and hence
\[
\mathbf c\in \Omega_P^{\mathcal M}(p).
\]
Since the cost realization \(\mathbf c\in\Omega_P^{\mathcal S}(p)\) was arbitrary, we have
\[
\Omega_P^{\mathcal S}(p)\subseteq \Omega_P^{\mathcal M}(p).
\]

When \(k=1\), there is only one backstopper, so the two protocols induce the same single-agent decision: the backstopper fills the deterministic gap \(D_K\) if and only if
\[
V\ge \Gamma^*(c_{(1)},p,D_K).
\]
Consequently, we show that the two provision regions are equivalent.

When \(k\ge2\), the inclusion is generically strict. The reason is that \(\mathcal M\)
uses the order of moves to shift the final residual gap to the lowest-cost backstopper,
whereas \(\mathcal S\) requires all backstoppers to commit simultaneously to their target
retentions. To see strictness formally, consider a heterogeneous pool with
\(c_{(1)}<c_{(2)}\). Define the maximum residual gap that type \(c\) can fill by
\[
\bar d(c;V,p)
=
\min\left\{1,\frac{p+\sqrt{2cV}}{c}\right\}.
\]
This bound is decreasing in \(c\). Hence there exists an open set of parameters for which
some residual gap \(R\) satisfies
\[
\bar d(c_{(2)};V,p)<R\le \bar d(c_{(1)};V,p).
\]
Under small-first, such a residual gap can be left to the final low-cost backstopper
\((1)\), who can fill it. Under simultaneous withdrawal, the same residual cannot be
shifted ex post to the lowest-cost user; all backstoppers must satisfy their target
constraints at once. For cost realizations in this open set, \(\mathcal M\) supports
provision while \(\mathcal S\) does not. Therefore the inclusion is strict generically when
multiple backstoppers are required.
\end{proof}

\subsection{Example}

We now present an example to illustrate a cost realisation for which $\mathcal M$ succeeds while $\mathcal S$ fails, and here we supply the complete calculations.

Consider $n = 3$, $X = 1.2005$, $p = 0.05$, $V = 3.5$, and cost vector $(c_{(1)}, c_{(2)}, c_{(3)}) = (10, 40, 100)$.  
The floor contributions are
\[
\begin{aligned}
e^0(c_{(1)})&=\frac{0.05}{10}=0.005,\\
e^0(c_{(2)})&=\frac{0.05}{40}=0.00125,\\
e^0(c_{(3)})&=\frac{0.05}{100}=0.0005.
\end{aligned}
\]
Because $e^0(c_{(1)}) + e^0(c_{(2)}) + e^0(c_{(3)}) = 0.00675 \ll X$, provision cannot be achieved by floor contributions alone.  The provider forms a backstop pool $K$ of the lowest‑cost users.  We check $k=1$: $D_1 = X - e^0(c_{(2)}) - e^0(c_{(3)}) \approx 1.2005 - 0.00125 - 0.0005 = 1.19875 > 1$, so one backstopper is insufficient.  For $k=2$, $K = \{1,2\}$,
\[
D_K = X - e^0(c_{(3)}) = 1.2005 - 0.0005 = 1.2,
\]
and $k-1 < D_K \le k$ holds.  Hence the provider assigns $g_{(1)} = g_{(2)} = 1$, $g_{(3)} = e^0(c_{(3)})$.

\medskip
\noindent\textbf{Outcome under $\mathcal M$.}
Under $\mathcal M$, backstopper $(2)$ (cost $40$) moves first, then backstopper $(1)$ (cost $10$) moves last.  Backstopper $(1)$ faces a residual gap $D_1$ and will fill it if and only if
\[
\Gamma^*(c_{(1)}, p, D_1) = \frac{(c_{(1)}D_1 - p)^2}{2c_{(1)}} \le V.
\]
Solving this inequality gives backstopper $(1)$'s maximum fillable gap
\[
\begin{aligned}
\bar D_1
&=\min\!\Bigl\{1,\;
\frac{p+\sqrt{2c_{(1)}V}}{c_{(1)}}\Bigr\}\\
&=\min\!\Bigl\{1,\;
\frac{0.05+\sqrt{20\times3.5}}{10}\Bigr\}\\
&=\min\!\Bigl\{1,\,0.84166\ldots\Bigr\}
\approx0.84166.
\end{aligned}
\]

Backstopper $(2)$ anticipates this behavior and chooses $e_{(2)} \in [p/c_{(2)}, 1]$ to solve
\[
\min_{e_{(2)}} \; \frac{c_{(2)}}{2}e_{(2)}^2 - p e_{(2)}
\quad\text{subject to}\quad D_K - e_{(2)} \le \bar D_1 .
\]
The objective is strictly decreasing in $e_{(2)}$ over the feasible range (since its derivative $c_{(2)}e_{(2)} - p$ is positive for $e_{(2)} > p/c_{(2)} = 0.00125$), so the constraint binds at the optimum.  Hence
\[
e_{(2)}^* = D_K - \bar D_1 = 1.2 - 0.84166 = 0.35834,
\]
and consequently $e_{(1)}^* = \bar D_1 \approx 0.84166$.  Rounding to two decimal places yields the retention profile $(e_{(1)}^{\mathcal M}, e_{(2)}^{\mathcal M}) = (0.84, 0.36)$ as reported in the main text.

The participation costs under these retentions are
\[
\begin{aligned}
\Gamma^*(10,0.05,0.84)
&=\frac{(10\times0.84-0.05)^2}{20}\\
&=\frac{(8.35)^2}{20}
\approx3.486<3.5,
\end{aligned}
\]
\[
\begin{aligned}
\Gamma^*(40,0.05,0.36)
&=\frac{(40\times0.36-0.05)^2}{80}\\
&=\frac{(14.35)^2}{80}
\approx2.574<3.5.
\end{aligned}
\]
Both constraints are satisfied, so provision succeeds under $\mathcal M$.

\medskip
\noindent\textbf{Outcome under $\mathcal S$.}
Under $\mathcal S$, the provider uses the equal‑marginal‑cost allocation.  The harmonic aggregate is
\[
\bar C_K = \frac{1}{10} + \frac{1}{40} = 0.125,
\]
and the individual targets are
\[
\begin{aligned}
e_{(1)}^{\mathcal S}
&=\frac{D_K}{\bar C_K \, c_{(1)}}
=\frac{1.2}{0.125\times10}=0.96,\\
e_{(2)}^{\mathcal S}
&=\frac{D_K}{\bar C_K \, c_{(2)}}
=\frac{1.2}{0.125\times40}=0.24.
\end{aligned}
\]
Both $e_{(1)}^{\mathcal S}=0.96$ and $e_{(2)}^{\mathcal S}=0.24$ lie in their respective feasible intervals $[e^0(c_j), 1]$ (since $0.96 \in [0.005,1]$ and $0.24 \in [0.00125,1]$), so the equal‑marginal‑cost allocation itself is feasible.  The provision BNE exists only if every backstopper's participation constraint is met.  For backstopper $(1)$,
\[
\begin{aligned}
\Gamma^*(10,0.05,0.96)
&=\frac{(10\times0.96-0.05)^2}{20}\\
&=\frac{(9.55)^2}{20}
\approx4.560>3.5.
\end{aligned}
\]
Because backstopper $(1)$ would not retain $0.96$ at $V = 3.5$, the provision BNE does not exist.  Hence $\mathcal S$ fails.  This example shows that the inclusion $\Omega_P^{\mathcal S}(p) \subset \Omega_P^{\mathcal M}(p)$ can be strict.

\subsection{Cost Efficiency in the Common Provision Region}
\label{app:cost-efficiency}

We now compare the total privacy costs when both withdrawal protocols succeed. By Proposition~\ref{prop:Containment}, this common region is simply \(\Omega_P^{\mathcal S}(p)\). The point of the comparison is to compare the intensive-margin cost of the successful retention profile.

Consider a fixed cost realization in the common provision region. Under both protocols, users outside the backstop pool \(K\) are assigned only their floor amounts and retain
\[
e_j=e^0(c_j),\qquad j\notin K.
\]
Therefore, their privacy costs are identical under \(\mathcal S\) and \(\mathcal M\). The only possible cost difference comes from the retentions of the backstoppers in \(K\). Once the non-backstoppers' floor contributions are fixed, the backstoppers must jointly supply the residual demand
\[
D_K
=
X-\sum_{j\notin K}e^0(c_j).
\]
Thus, any successful backstopper retention vector must provide at least \(D_K\) in total.
For cost-efficiency, however, there is no reason to provide strictly more than \(D_K\):
the model-improvement benefit is already obtained once the threshold \(X\) is reached,
while privacy cost
\[
C_K(e)=\sum_{j\in K}\frac{c_j e_j^2}{2}
\]
strictly increases with each positive retention \(e_j\). Hence any cost-minimizing
successful profile satisfies the exact equality
\[
\sum_{j\in K}e_j=D_K.
\]

Each backstopper \(j\in K\) must also choose a feasible retention level. He cannot retain
less than his floor in a provision path, because the floor \(e^0(c_j)\) is privately optimal
from the subsidy alone; and he cannot retain more than one unit by normalization. Hence
the relevant feasible set for the cost comparison is
\[
\begin{aligned}
\mathcal E_K(D_K)
&=
\left\{
e=(e_j)_{j\in K}:
e_j\in[e^0(c_j),1],\right.\\
&\hspace{4.5em}\left.
\forall j\in K,\quad
\sum_{j\in K}e_j=D_K
\right\}.
\end{aligned}
\]

\begin{prop}[Cost efficiency of \(\mathcal S\) on the common provision region]
\label{prop:cost-efficiency}
For any cost realization
\(\mathbf c\in\Omega_P^{\mathcal S}(p)\), let \(e^{\mathcal S}\) be the simultaneous
target-retention vector and let \(e^{\mathcal M}\) be the small-first equilibrium
retention vector. Then
\[
\sum_{j\in K}\frac{c_j (e_j^{\mathcal S})^2}{2}
\le
\sum_{j\in K}\frac{c_j (e_j^{\mathcal M})^2}{2}.
\]
The inequality is strict whenever the small-first retention vector differs from the
cost-minimizing target and no feasibility bound makes the minimizer non-unique.
\end{prop}

\begin{proof}
We first characterize the cost-minimizing successful retention vector. Consider the
convex program
\[
\min_{e\in\mathcal E_K(D_K)}
\sum_{j\in K}\frac{c_j e_j^2}{2}.
\]
The objective is strictly convex, and the feasible set is convex. Therefore the solution is
unique whenever the feasible set has a nonempty relative interior.

In the interior case where no bound \(e_j\in[e^0(c_j),1]\) is binding, the Lagrangian is
\[
\mathcal L(e,\lambda)
=
\sum_{j\in K}\frac{c_j e_j^2}{2}
+
\lambda\left(D_K-\sum_{j\in K}e_j\right).
\]
The first-order condition for each \(j\in K\) is
\[
c_j e_j - \lambda = 0 \;\Rightarrow\; c_j e_j = \lambda \quad (\forall j\in K).
\]
the cost-minimizing allocation equalizes marginal privacy costs. Since
\[
\sum_{j\in K}e_j
=
\sum_{j\in K}\frac{\lambda}{c_j}
=
\lambda \sum_{j\in K}\frac{1}{c_j}
=
D_K,
\]
we obtain
\[
\lambda = \frac{D_K}{\bar C_K},\quad 
\bar C_K = \sum_{\ell\in K}\frac{1}{c_\ell}
\;\Rightarrow\;
e_j^{\mathcal S} = \frac{D_K}{\bar C_K c_j}.
\]

This is exactly the simultaneous target-retention vector used in the main. When some feasibility bounds bind, the same cost-minimization logic gives a constrained version of the target. To see this, write the lower and upper bounds as
\[
\ell_j=e^0(c_j),\qquad u_j=1.
\]
The constrained problem is
\[
\begin{aligned}
\min_{\{e_j\}_{j\in K}}
&\quad\sum_{j\in K}\frac{c_j e_j^2}{2}\\
\text{s.t.}
&\quad\sum_{j\in K}e_j=D_K,\\
&\quad\ell_j\le e_j\le u_j.
\end{aligned}
\]
The KKT conditions imply that, for every unconstrained backstopper $c_j e_j=\lambda$. Thus the unconstrained users still equalize marginal privacy costs. Users whose unconstrained target \(\lambda/c_j\) falls below the floor are fixed at the lower bound \(e^0(c_j)\), while users whose unconstrained target exceeds one are fixed at the upper
bound \(1\). Equivalently, the constrained optimum has the water-filling form
\[
e_j^{\mathcal S}
=
\min\left\{
1,\,
\max\left\{
e^0(c_j),\frac{\lambda}{c_j}
\right\}
\right\},
\]
where \(\lambda\) is chosen so that
\[
\sum_{j\in K}e_j^{\mathcal S}=D_K.
\]
Therefore, whether the interior formula applies or some bounds bind,
\(e^{\mathcal S}\) is the cost-minimizing successful retention vector in
\(\mathcal E_K(D_K)\).

Now consider Protocol~\(\mathcal M\) on the same cost realization. Since
\[
\mathbf c\in\Omega_P^{\mathcal S}(p)\subseteq\Omega_P^{\mathcal M}(p),
\]
the small-first protocol also succeeds. Its equilibrium retention vector \(e^{\mathcal M}\) must cover the same residual demand \(D_K\). If it over-contributes strictly above \(D_K\), reducing some positive retention until the aggregate equals \(D_K\) would keep provision successful and lower total privacy cost. Therefore, for the purpose of cost comparison, the relevant successful small-first retention vector belongs
to the same feasible set:
\[
e^{\mathcal M}\in\mathcal E_K(D_K).
\]
Because \(e^{\mathcal S}\) minimizes \(C_K(e)\) over \(\mathcal E_K(D_K)\), we have
\[
C_K(e^{\mathcal S})
\le
C_K(e^{\mathcal M}),
\]
which proves the desired inequality.

The inequality is strict whenever \(e^{\mathcal M}\neq e^{\mathcal S}\) and the cost-minimizing vector is unique. \(\mathcal S\) directly implements the least-cost division of the residual demand among backstoppers. Protocol~\(\mathcal M\), in contrast, uses the withdrawal order to expand the set of cost realizations where provision can be sustained. This gives \(\mathcal M\) a larger provision region, while \(\mathcal S\) is more cost-efficient conditional on both protocols succeeding.
\end{proof}

\section{ Supplementary Material for experiments}
\subsection{Details of the Numerical Implementation}
\label{app:num-details}

We set $n=50$, $X=10.5$, and draw privacy costs independently from $U[1,5]$. 
The parameters $V$ and $p$ vary on an $80\times80$ grid over $[0,5]\times[0,0.65]$. 
For each $(V,p)$ pair, we draw $N_{\mathrm{mc}}=5{,}000$ independent cost vectors and compute the equilibrium outcome of each mechanism. 
Provision probability is the fraction of simulations in which aggregate contributions reach $X$, and expected welfare follows equation~\ref{eq:EW} of the main text.

\medskip
\noindent\textbf{Monte Carlo uncertainty.}
The equilibrium outcome of each mechanism is deterministic conditional on a realised cost vector. 
However, the reported provision probabilities and expected welfare values are Monte Carlo estimates over independently sampled cost vectors. 
For a mechanism $J\in\{\mathcal C,\mathcal S,\mathcal M\}$ and grid point $(V,p)$, let
$Y_{r}^{J}(V,p)\in\{0,1\}$ denote the provision-success indicator in Monte Carlo draw $r$. 
We estimate the provision probability by
\[
\widehat{\rho}^{J}(V,p)
=
\frac{1}{N_{\mathrm{mc}}}
\sum_{r=1}^{N_{\mathrm{mc}}}
Y_{r}^{J}(V,p),
\]
with standard error
\[
\mathrm{SE}\!\left(\widehat{\rho}^{J}(V,p)\right)
=
\sqrt{
\frac{
\widehat{\rho}^{J}(V,p)
\left(1-\widehat{\rho}^{J}(V,p)\right)
}
{N_{\mathrm{mc}}}
}.
\]
The pointwise 95\% confidence interval is computed as
\[
\widehat{\rho}^{J}(V,p)
\pm
1.96\,
\mathrm{SE}\!\left(\widehat{\rho}^{J}(V,p)\right).
\]
Since $N_{\mathrm{mc}}=5{,}000$, the worst-case standard error for a Bernoulli success-probability estimate is
\[
\sqrt{0.25/5000}\approx 0.0071,
\]
corresponding to a worst-case pointwise 95\% confidence half-width of approximately $0.014$.

For expected welfare, let $W_{r}^{J}(V,p)$ denote the realised social welfare of mechanism $J$ in draw $r$, computed according to equation ~\ref{eq:EW} of the main text. 
We estimate expected welfare by
\[
\widehat{SW}^{J}(V,p)
=
\frac{1}{N_{\mathrm{mc}}}
\sum_{r=1}^{N_{\mathrm{mc}}}
W_{r}^{J}(V,p),
\]
and report its Monte Carlo standard error as
\[
\mathrm{SE}\!\left(\widehat{SW}^{J}(V,p)\right)
=
\frac{s_{W}^{J}(V,p)}{\sqrt{N_{\mathrm{mc}}}},
\]
where $s_{W}^{J}(V,p)$ is the sample standard deviation of 
$\{W_{r}^{J}(V,p)\}_{r=1}^{N_{\mathrm{mc}}}$.

\medskip
\noindent\textbf{Mechanism $\mathcal C$: equilibrium selection.}
For a given belief parameter $b_0$, we first search for a supported cutoff equilibrium. 
A grid of candidate participation probabilities $q\in(0,b_0]$ is evaluated: for each $q$, we set $a=F^{-1}(q)$, compute $\tilde g(a)$ from the cutoff-equilibrium expression in the main text, and estimate the net gain 
\[
\Phi(a;V,p)=V\Delta B(a)-\Gamma^*(a,p,\tilde g(a))
\]
using $B_{\mathrm{mc}}$ auxiliary draws of $n-1$ cost types. 
If any candidate satisfies $\Phi\ge0$, the one with the largest $\Phi$ is selected; otherwise only the floor equilibrium remains. 
We then simulate $N_{\mathrm{mc}}$ realised cost vectors, applying the selected cutoff strategy when available and the floor strategy otherwise. 
Algorithm~\ref{alg:C} summarises the procedure.

\begin{algorithm}[ht]
\caption{Mechanism $\mathcal C$ for a given $(V,p,b_0)$}
\label{alg:C}
\begin{algorithmic}[1]
\State Search $q\in(0,b_0]$: select $a^*$ with largest $\Phi(a;V,p)\ge0$, or null if none exists.
\For{$t=1$ \textbf{to} $N_{\mathrm{mc}}$}
    \State Draw $\mathbf{c}\sim F^n$; compute floor contributions $\mathbf{e}^0$.
    \If{$\sum_i e^0_i\ge X$}
        \State success $\gets$ true
        \State Set $\mathbf{e}\gets\mathbf{e}^0$.
    \ElsIf{$a^*$ exists}
        \State Apply cutoff $a^*$ to obtain $\mathbf{e}$.
        \State success $\gets \left(\sum_i e_i\ge X\right)$
    \Else
        \State success $\gets$ false
        \State Set $\mathbf{e}\gets\mathbf{e}^0$.
    \EndIf
    \State Accumulate success and welfare.
\EndFor
\end{algorithmic}
\end{algorithm}

\medskip
\noindent\textbf{Mechanism $\mathcal S$: simultaneous withdrawal.}
For each cost realisation, we determine the minimal backstop pool $K$ via equation~\eqref{eq:dk}. 
The equal-marginal-cost targets $e_j^*$ are computed as in equation~\eqref{eq:S-best-alloc}, and provision succeeds exactly when every backstopper's participation constraint in equation~\eqref{eq:SV} is satisfied, following the construction in the main text. 
The mechanism is deterministic given the cost realisation and requires no equilibrium-selection parameters.

\medskip
\noindent\textbf{Mechanism $\mathcal M$: small-first withdrawal.}
Using the same backstop pool $K$, backstoppers move in the order specified by the small-first protocol. 
We implement the backward-induction procedure: the chain-success indicators $S_j$ are defined recursively as in equations~\eqref{eq:S1} and~\eqref{eq:regin2}, and provision succeeds if and only if the full chain succeeds from the initial floor aggregate $E_F$. 
The mechanism is also deterministic given the cost realisation and requires no equilibrium-selection parameters.

\medskip
\noindent\textbf{Paired mechanism comparisons.}
All mechanisms are evaluated on the same Monte Carlo cost draws at each grid point. 
Therefore, when comparing two mechanisms $J$ and $J'$, we compute paired differences rather than treating the two estimates as independent. 
For provision probability, define
\[
D_{r}^{J,J'}(V,p)
=
Y_{r}^{J}(V,p)-Y_{r}^{J'}(V,p).
\]
The estimated difference in provision probability is
\[
\widehat{\Delta}_{\rho}^{J,J'}(V,p)
=
\frac{1}{N_{\mathrm{mc}}}
\sum_{r=1}^{N_{\mathrm{mc}}}
D_{r}^{J,J'}(V,p),
\]
with standard error
\[
\mathrm{SE}\!\left(\widehat{\Delta}_{\rho}^{J,J'}(V,p)\right)
=
\frac{s_{D}^{J,J'}(V,p)}{\sqrt{N_{\mathrm{mc}}}},
\]
where $s_{D}^{J,J'}(V,p)$ is the sample standard deviation of the paired differences 
$\{D_{r}^{J,J'}(V,p)\}_{r=1}^{N_{\mathrm{mc}}}$. 
The same paired-difference calculation is used for welfare comparisons, replacing 
$Y_{r}^{J}(V,p)$ by $W_{r}^{J}(V,p)$.

\subsection{Robustness Checks}
\label{app:robustness}

We now conduct two sets of robustness exercises. The first replaces the uniform cost distribution by skewed Beta‑mixture distributions; the second introduces noisy cost observation for mechanisms $\mathcal{S}$ and $\mathcal{M}$, relaxing the assumption that the provider perfectly observes privacy costs and users exhibit no potential misreporting behaviors.

\subsubsection{Alternative cost distribution}

In the baseline simulations, costs are drawn i.i.d.\ from $U[1,5]$.  We now replace this by a Beta distribution rescaled to the interval $[1,5]$, which maintains the same support while allowing the density to be either left‑skewed or right‑skewed.
For left‑skewed costs we use $\text{Beta}(2,5)$, and for right‑skewed costs $\text{Beta}(5,2)$.  Both specifications satisfy the log‑concavity requirement assumed in the theoretical analysis, and they preserve the support $[c_{\text{low}},c_{\text{high}}]=[1,5]$, so the only change is the shape of the distribution.

For each skew direction we repeat the main experiment on an $80\times 80$ $(V,p)$ grid with $N_{\text{mc}}=10\,000$ cost vectors per cell and $B_{\text{mc}}=10\,000$ auxiliary draws for the $\mathcal{C}$‑mechanism belief search.  Mechanism $\mathcal{C}$ is evaluated with the intermediate belief $b_0=0.15$, exactly as in the main text. Figures~\ref{fig:robust-left-combined} and~\ref{fig:robust-right-combined} report success probabilities and expected social welfare for the three mechanisms under left‑skewed and right‑skewed costs, respectively.

\begin{figure}[ht]
    \centering
    \begin{subfigure}[b]{.90\linewidth}
        \centering
        \includegraphics[width=\linewidth]{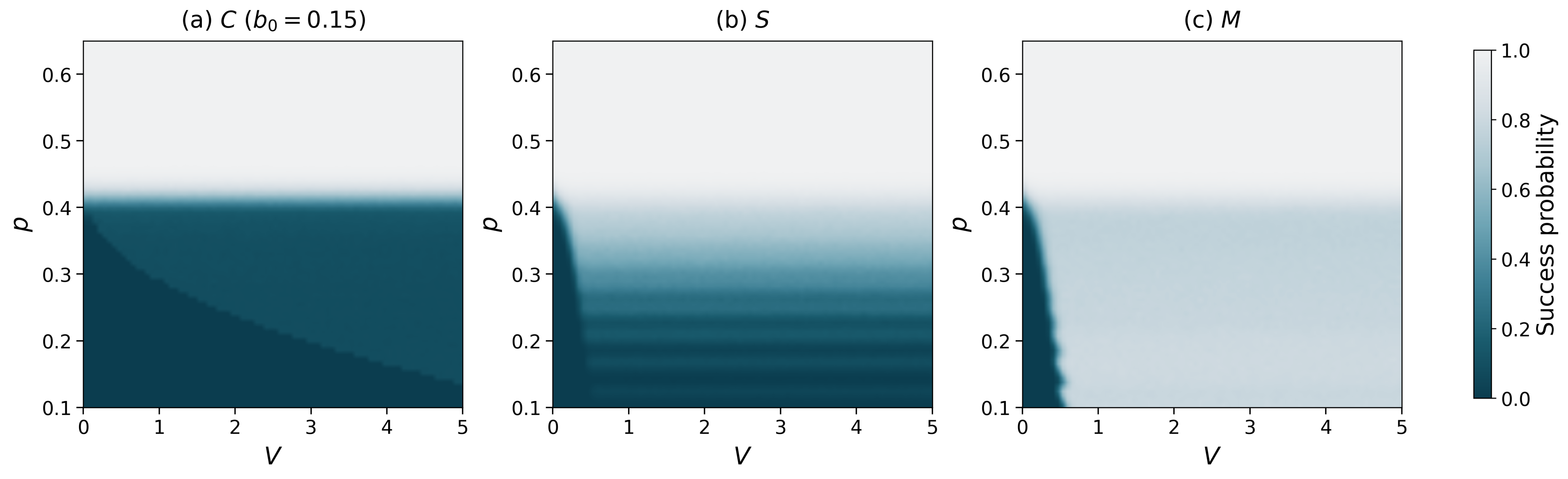}
        \label{fig:robust-left-succ}
    \end{subfigure} 
    \vspace{1em} 
    \begin{subfigure}[b]{.90\linewidth}
        \centering
        \includegraphics[width=\linewidth]{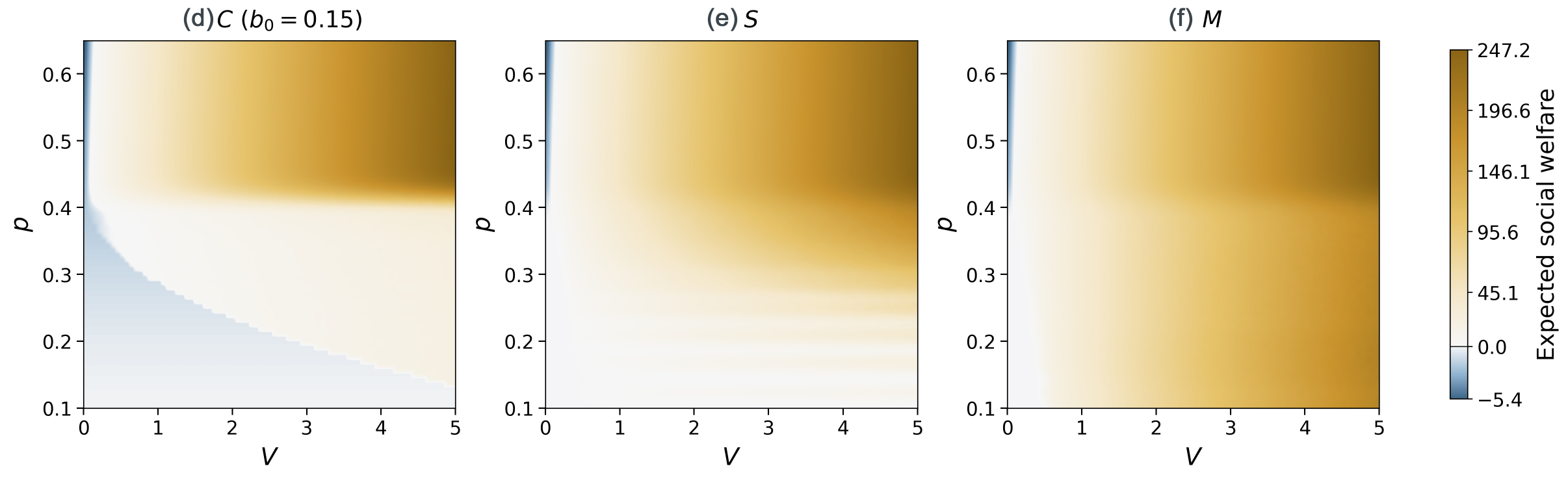}
        \label{fig:robust-left-welf}
    \end{subfigure}
    \caption{Results under left‑skewed costs ($\text{Beta}(2,5)$, mode at low $c$). Top panel: provision success probability for mechanisms $\mathcal{C}$ ($b_0=0.15$), $\mathcal{S}$, and $\mathcal{M}$.  Bottom panel: corresponding expected social welfare.}
    \label{fig:robust-left-combined}  
\end{figure}

\begin{figure}[ht]
    \centering
    \begin{subfigure}[b]{.90\linewidth}
        \centering
        \includegraphics[width=\linewidth]{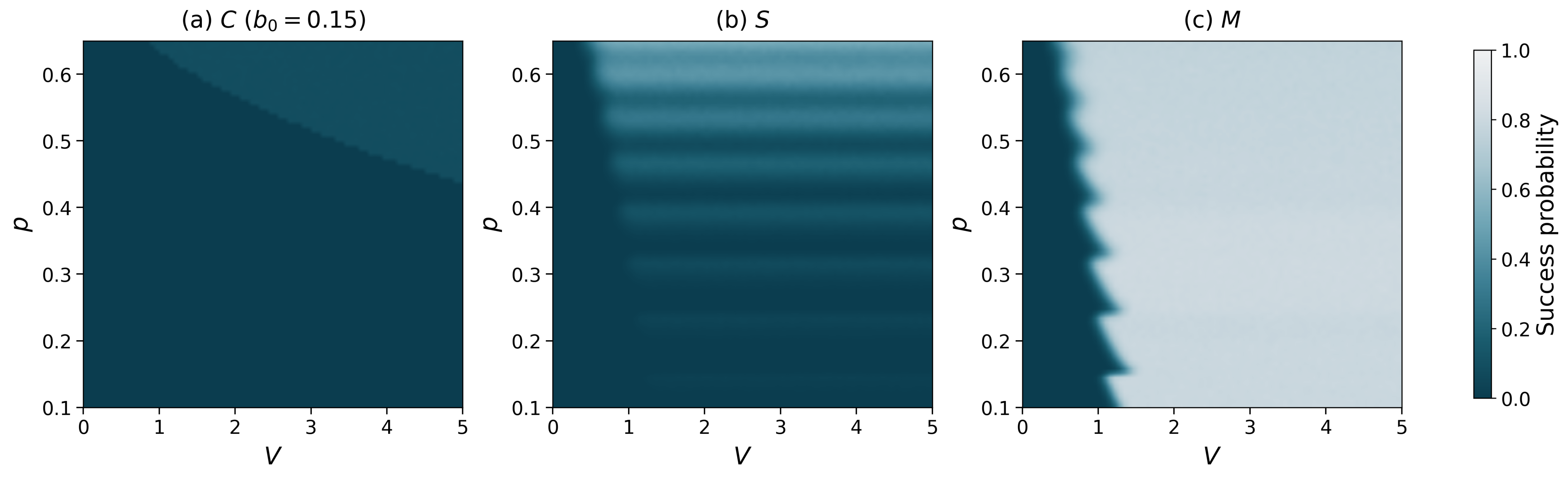}
        \label{fig:robust-right-succ}
    \end{subfigure}
    \vspace{1em}
    \begin{subfigure}[b]{.90\linewidth}
        \centering
        \includegraphics[width=\linewidth]{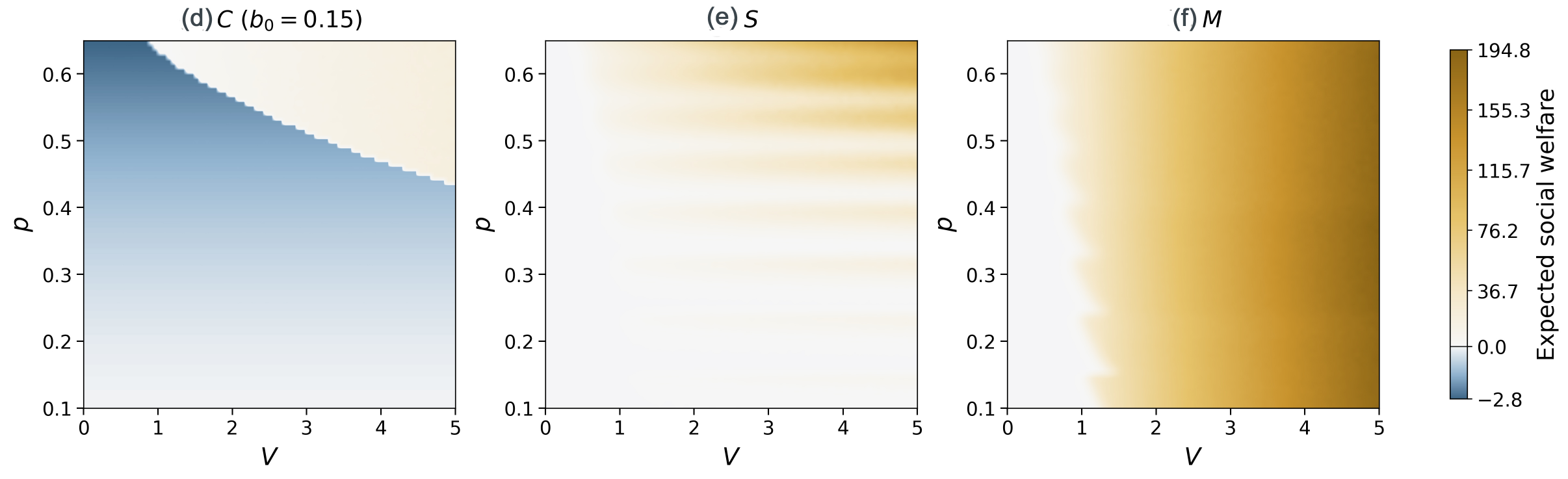}
        \label{fig:robust-right-welf}
    \end{subfigure}
    \caption{Results under left‑skewed costs ($\text{Beta}(5,2)$, mode at high $c$). Top panel: provision success probability for mechanisms $\mathcal{C}$ ($b_0=0.15$), $\mathcal{S}$, and $\mathcal{M}$.  Bottom panel: corresponding expected social welfare.}
    \label{fig:robust-right-combined}
\end{figure}

Under left‑skewed costs (i.e., when low‑cost users are more abundant), the probability of successful model improvement increases substantially for all three mechanisms.  At the same time, the region of negative expected welfare expands slightly in the high‑$p$, low‑$V$ corner.  There, the improvement value $V$ is too small to offset privacy costs, and the higher success rate under the left‑skewed distribution actually deepens the welfare loss.  In contrast, right‑skewed costs (fewer low‑cost users) shrink the provision region and, under mechanism $\mathcal{C}$, produce a noticeably larger subsidy‑leakage area.

Despite these distributional shifts, the central findings of the paper remain unchanged.  In particular,
\begin{itemize}
    \item Both withdrawal mechanisms eliminate the negative‑welfare region that plagues mechanism $\mathcal{C}$ under the floor equilibrium;
    \item The welfare ranking $\mathcal{M}>\mathcal{S}>\mathcal{C}$ is preserved across all cost distributions considered.
\end{itemize}
These results confirm that the paper's conclusions do not hinge on the uniform distribution and are robust to substantial changes in the shape of the cost distribution, as long as the density remains log‑concave.

\subsubsection{Noisy observation of privacy costs}

A strong assumption underlying the baseline withdrawal mechanisms is that the provider can base its assignment rule on accurate cost information. This assumption concerns the information available to the provider after users disclose their costs; it is distinct from the strategic issue of whether users have incentives to misreport their costs. In this subsection, we focus on the former issue and examine whether the main comparison between the withdrawal protocols is robust when the provider observes only noisy cost signals.

Specifically, let the provider observe
\[
    \hat c_i = c_i \exp(\eta_i),
    \qquad
    \eta_i \sim \mathcal{N}(0,\tau^2),
\]
where $\eta_i$ is independent across users and independent of $c_i$. The provider determines the backstop pool and assignments based on $\{\hat c_i\}$, treating these signals as the relevant cost estimates, while users make withdrawal decisions according to their true costs $c_i$. Thus, noise affects the provider's allocation rule but does not change users' underlying privacy costs or their withdrawal incentives.

We consider two levels of noise, $\tau=0.2$ and $\tau=0.5$, and compare them with the noiseless benchmark. Mechanism~$\mathcal C$ is unaffected because it does not require the provider to observe users' costs. Therefore, the robustness exercise focuses on the two withdrawal mechanisms, $\mathcal S$ and $\mathcal M$. Figure~\ref{fig:robust-noisy} reports the provision success probabilities over the same $(V,p)$ grid as in the main experiment.

\begin{figure}[ht]
    \centering
    \includegraphics[width=\linewidth]{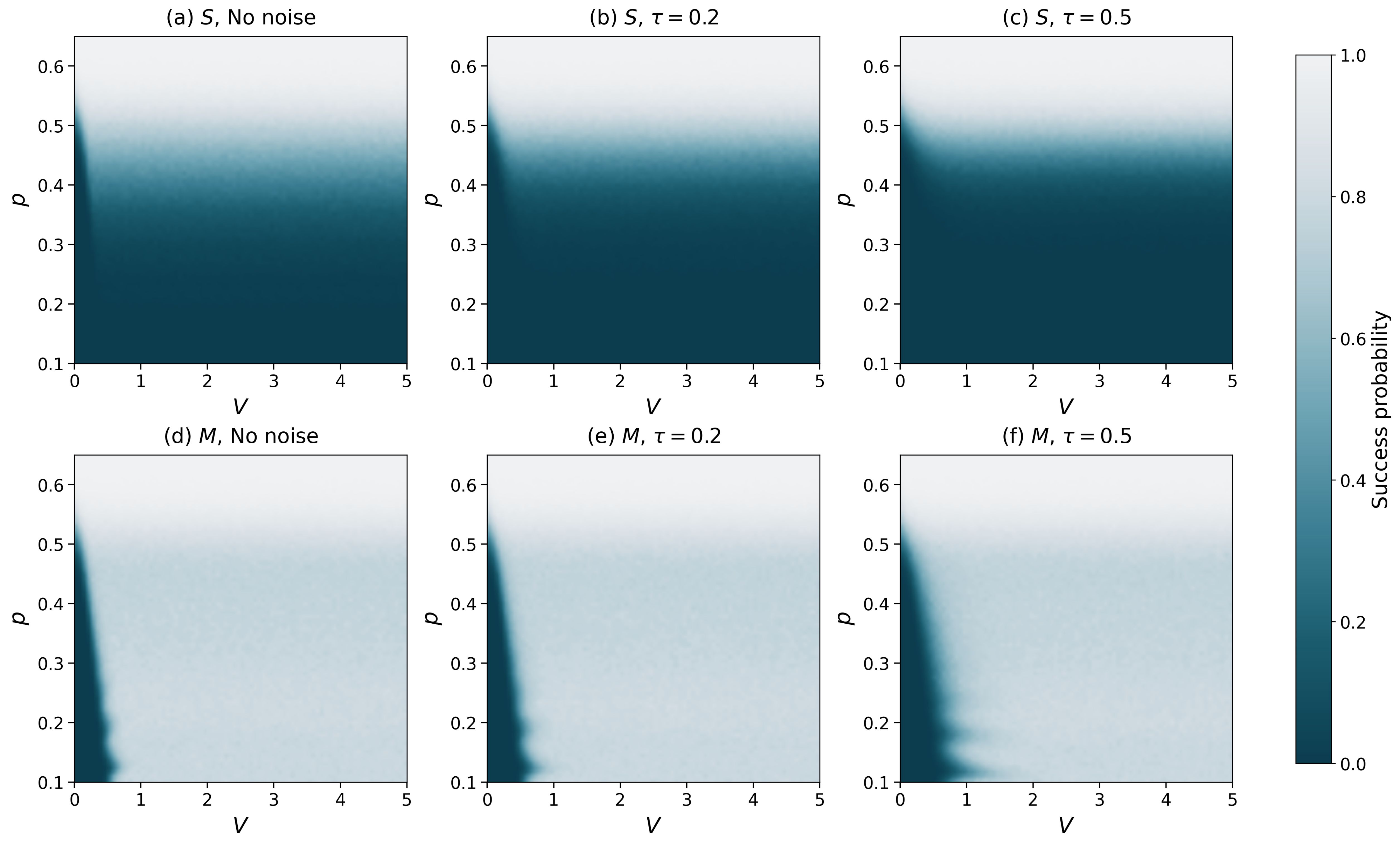}
    \caption{Success probabilities under noisy cost observation. The first row reports mechanism~$\mathcal S$, and the second row reports mechanism~$\mathcal M$. The columns correspond to the noiseless benchmark, moderate noise $(\tau=0.2)$, and substantial noise $(\tau=0.5)$. Noisy cost observation reduces the success region of the simultaneous protocol~$\mathcal S$, whereas the sequential protocol~$\mathcal M$ is mainly affected near the provision boundary and remains largely stable in the interior of its success region.}
    \label{fig:robust-noisy}
\end{figure}

\textbf{Result}. Cost-observation noise weakens the performance of the withdrawal mechanisms, but the magnitude of this effect differs substantially across protocols. For the simultaneous protocol~$\mathcal S$, noise leads to a visible reduction in provision success over a relatively broad region of the $(V,p)$ space. By contrast, mechanism~$\mathcal M$ is considerably more robust. The noisy signals mainly affect observations close to the success boundary, where assignment errors can change whether the backward-induction chain remains feasible. Away from this boundary, the success region of~$\mathcal M$ is almost unchanged. This result demonstrates that even in non-ideal situations with noisy cost observations, the improved success probability of mechanism ~$\mathcal M$ still outperforms that of mechanism ~$\mathcal S$ in most regions of the parameter space.

\subsection{Cost Efficiency and Pointwise Pareto analysis}
\label{app:cost & pareto}
\subsubsection{Masked cost-efficiency comparison between \texorpdfstring{$\mathcal{S}$ and $\mathcal{M}$}{S and M}.}

The main experiment reports expected welfare, which combines two distinct effects: the small-first protocol $\mathcal{M}$ may succeed for a wider range of cost realisations than the simultaneous protocol $\mathcal{S}$, and, even when both protocols succeed, the two arrangements may generate different privacy costs.  This subsection separates these two channels.  The purpose is not to repeat the overall success-probability comparison, but to isolate the multi-backstopper cases in which the two withdrawal protocols diverge and then compare their conditional privacy costs.

We draw independent cost vectors for each pair $(V,p)$ and compute the equilibrium outcomes under both withdrawal protocols. Let $s^J(c;V,p)$ denote the success indicator of protocol $J\in\{\mathcal{S},\mathcal{M}\}$, and let $k(c,p)$ denote the size of the backstop pool selected by the provider. The event $k(c,p)\geq 2$ identifies cost realizations in which multiple backstoppers are required. According to Section ~\ref{Single-Backstopper Case}, this restriction is important because the two protocols coincide in the single-backstopper case. 

\begin{figure}[t]
    \centering
    \includegraphics[width=\linewidth]{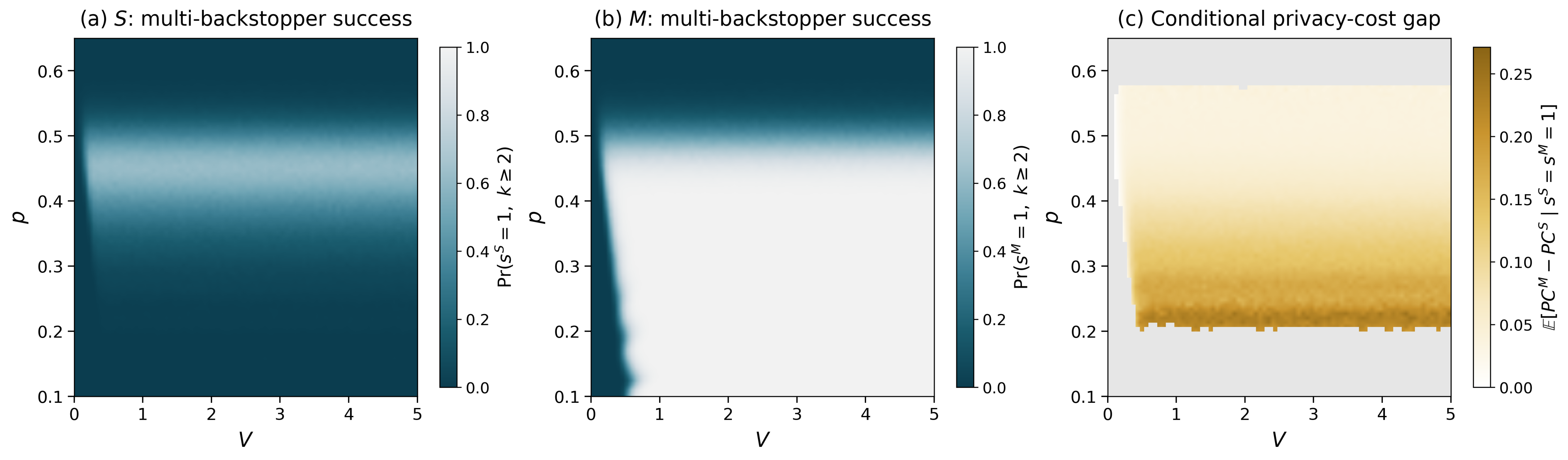}
    \caption{
    Masked cost-efficiency for the $\mathcal{S}$ and $\mathcal{M}$.
    Panel (a) reports the probability that $\mathcal{S}$ succeeds in multi-backstopper instances, $\Pr(s^\mathcal{S}=1,k\geq 2)$.
    Panel (b) reports the corresponding probability for $\mathcal{M}$, $\Pr(s^\mathcal{M}=1,k\geq 2)$.
    Panel (c) reports the conditional privacy-cost gap $PC^\mathcal{M}-PC^\mathcal{S}$ on the common-success set.
    Gray regions in Panel (c) indicate grid points with too few common-success samples.
    }
    \label{fig:d3-cost-efficiency}
\end{figure}

\textbf{Result}. Figure~\ref{fig:d3-cost-efficiency} reports the diagnostic results. Panels (a) and (b) show that the difference between the two protocols is concentrated in the multi-backstopper region. The simultaneous protocol $\mathcal{S}$ succeeds in a relatively narrow part of this region, mainly where the subsidy is high enough to reduce the remaining backstop burden. By contrast, the small-first protocol $\mathcal{M}$ succeeds over a much larger multi-backstopper region.

Panel (c) then masks out all cost realizations in which at least one protocol fails and compares privacy costs only on the common-success set,
\[
\mathcal{C}_{\mathcal{S}\mathcal{M}}(V,p)
=
\{c:s^\mathcal{S}(c;V,p)=1,\ s^\mathcal{M}(c;V,p)=1\}.
\]
It reports the average total privacy-cost gap $PC^\mathcal{M}-PC^\mathcal{S}$ over this set.

The conditional cost gap is nonnegative throughout the valid comparison region and the positive gap is concentrated in the parameter region where multi-backstopper cases are relevant. This pattern is consistent with the theoretical cost-efficiency result: when both protocols succeed, users outside the backstop pool retain the same floor contributions, so any cost difference comes from how the residual burden is allocated among backstoppers. The simultaneous protocol $\mathcal{S}$ uses the equal-marginal-cost allocation and therefore minimizes total privacy cost among successful retention profiles. The small-first protocol $\mathcal{M}$ instead uses sequencing to enlarge the set of feasible provision outcomes, but this comes at the cost of a less cost-minimizing allocation of the residual burden.

The numerical results confirm the efficiency--robustness trade-off. Protocol $\mathcal{M}$ is more robust because it succeeds in more heterogeneous multi-backstopper instances. Conditional on both protocols succeeding, however, protocol $\mathcal{S}$ remains more cost-efficient. Equivalently, $\mathcal{M}$'s welfare advantage in the main experiment is driven by an extensive-margin effect---a larger provision region---rather than by lower privacy cost conditional on success.

\subsubsection{Pointwise Pareto diagnostics.}

The welfare comparison in the main text is based on expected social welfare. This aggregate criterion is useful for comparing mechanisms, but it does not imply that every user is weakly better off in every realized cost vector. To distinguish aggregate welfare improvement from individual-level improvement, we define a pointwise Pareto improvement of mechanism $\mathcal{M}$ over mechanism $\mathcal{C}$ at a realized cost vector $\mathbf{c}$ as
\[
    u_i^{\mathcal{M}}(\mathbf{c};V,p)
    \ge
    u_i^{\mathcal{C}}(\mathbf{c};V,p),
    \qquad \forall i\in N,
\]
with strict inequality for at least one user when a strict Pareto improvement is considered.

We run a pointwise Pareto diagnostic to examine whether the welfare gain of the small-first protocol $\mathcal{M}$ can be interpreted as a user-side Pareto improvement over the subsidy-only mechanism $\mathcal{C}$. At each parameter value $(V,p)$, we draw independent cost vectors and compute each user's realized payoff under $\mathcal{C}$ and $\mathcal{M}$ using the same realization. We report three statistics:
\[
    \Pr_{\mathbf{c}}\!\left[
    u_i^{\mathcal{M}}(\mathbf{c};V,p)
    \ge
    u_i^{\mathcal{C}}(\mathbf{c};V,p),\ \forall i
    \right],
\]
the average share of users who are worse off,
\[
    \mathbb{E}_{\mathbf{c}}\left[
    \frac{1}{n}\sum_i
    \mathbf{1}\{u_i^{\mathcal{M}}(\mathbf{c};V,p)
    <
    u_i^{\mathcal{C}}(\mathbf{c};V,p)\}
    \right],
\]
and the average per-user compensation needed to eliminate all individual losses,
\[
    \mathbb{E}_{\mathbf{c}}\left[
    \frac{1}{n}\sum_i
    \max\{u_i^{\mathcal{C}}(\mathbf{c};V,p)
    -
    u_i^{\mathcal{M}}(\mathbf{c};V,p),0\}
    \right].
\]

\begin{figure}[t]
    \centering
    \includegraphics[width=\linewidth]{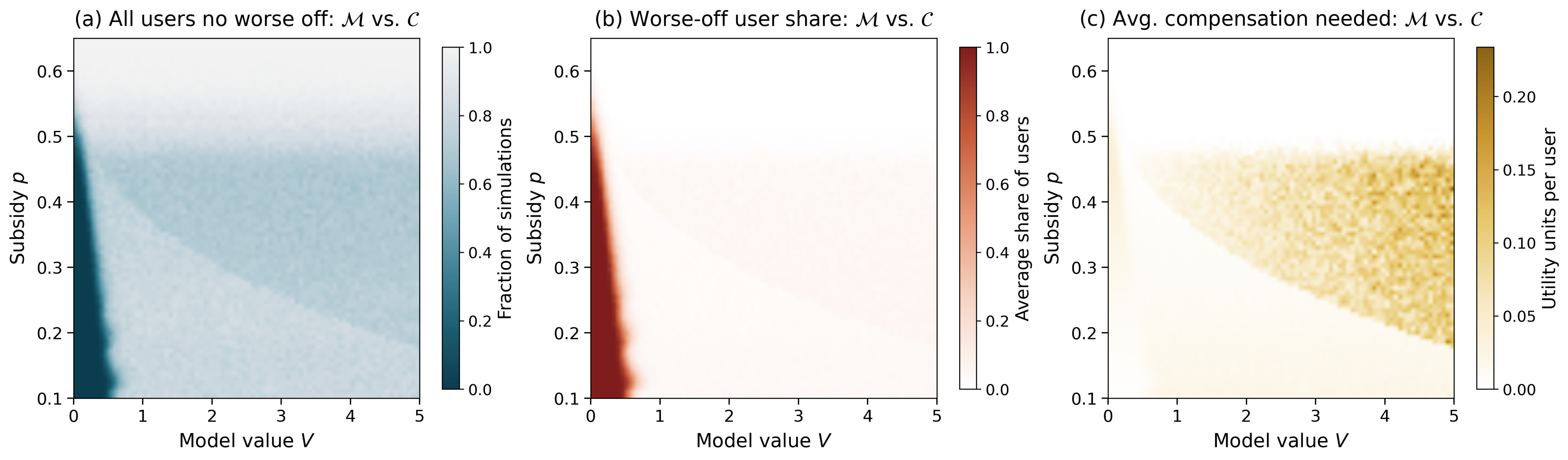}
    \caption{
    Pointwise Pareto diagnostics comparing the small-first withdrawal protocol $\mathcal{M}$ with the subsidy-only mechanism $\mathcal{C}$.
    The three panels report the fraction of simulations in which no user is worse off, the average share of worse-off users, and the average per-user compensation needed to make all users weakly better off.
    }
    \label{fig:d4-pareto}
\end{figure}

\textbf{Result}. Figure~\ref{fig:d4-pareto} shows that $\mathcal{M}$ does not uniformly dominate $\mathcal{C}$ at the individual-payoff level. The main regions where some users are worse off correspond to cases in which $\mathcal{C}$ still gives users favorable payoffs: either users obtain subsidy-induced floor payoffs even though provision fails, or the cutoff equilibrium under $\mathcal{C}$ already supports provision with a relatively favorable burden allocation. In these cases, switching to $\mathcal{M}$ may remove some floor-payoff exposure or reallocate the residual contribution burden toward backstoppers, so a small fraction of users can receive lower realized payoffs.

The heatmaps also show that these losses are limited. The harmed-user region is concentrated mainly in low-$V$ or intermediate-subsidy areas, while most of the parameter space is close to zero. The compensation gap is positive only in a bounded region and remains modest, mostly between $0$ and $0.2$ utility units per user. Thus, $\mathcal{M}$ should be interpreted as a potential Pareto improvement rather than an unconditional pointwise Pareto improvement: it expands provision and raises aggregate surplus, but a small additional transfer may be needed to ensure that no individual user is worse off.

\end{document}